\documentclass[lettersize,journal]{IEEEtran}
\usepackage{amsmath,amsfonts}
\usepackage{array}
\usepackage{multirow}
\usepackage{textcomp}
\usepackage{stfloats}
\usepackage{hyperref}
\usepackage{graphicx}
\usepackage{cite}
\usepackage{amssymb}
\usepackage{amsmath}
\usepackage{amsthm}
\usepackage{babel}
\usepackage{url}
\usepackage{braket}
\usepackage{algorithm}
\usepackage{algorithmicx}
\usepackage{algpseudocode}

\usepackage{comment}

\usepackage{tikz}
\usepackage{circuitikz}
\usepackage{qcircuit}

\theoremstyle{plain}

\newtheorem{theorem}{Theorem}
\newtheorem{lemma}{Lemma}
\newtheorem{remark}{Remark}

\theoremstyle{definition}
\newtheorem{definition}{Definition}

\newcommand{\tket}{$\text{t}\!\ket{\mathrm{ket}}$}
\newcommand{\sabre}{\textsf{SABRE}}
\newcommand{\mcts}{\textsf{MCTS}}
\newcommand{\sahs}{\textsf{SAHS}}
\newcommand{\fidls}{\textsf{FiDLS}}

\newcommand{\ag}{\mathbb{AG}}
\newcommand{\V}{\mathbb{V}}
\newcommand{\E}{\mathbb{E}}
\newcommand{\qknob}{\texttt{QKNOB}}
\newcommand{\queko}{\texttt{QUEKO}}

\newcommand{\swap}[2]{\textsc{swap}({#1},{#2})}

\newcommand{\define}{\ensuremath{\triangleq}}
\newcommand{\la}{\langle}
\newcommand{\ra}{\rangle}
\newcommand{\cx}[2]{\langle{#1,#2}\rangle}
\newcommand{\glink}[3]{{#1}\xrightarrow{#2}{#3}}

\newcommand{\swapnorm}[1]{\left\lVert#1\right\rVert}

\makeatletter
\newcommand{\xRightarrow}[2][]{\ext@arrow 0359\Rightarrowfill@{#1}{#2}}

\newcommand{\sglink}[3]{{#1}\xRightarrow{#2}{#3}}

\usepackage[normalem]{ulem}

\usepackage{stmaryrd} 
\newcommand{\sem}[1]{\llbracket #1 \rrbracket}

\usepackage{mathtools}
\usepackage{booktabs}   

\begin{document}

\title{Benchmarking Quantum Circuit Transformation with QKNOB Circuits\thanks{Work partially supported by the Australian Research Council (DP220102059) and the National Science Foundation of China (12071271).}
}

\author{Sanjiang~Li, \thanks{Sanjiang Li is with Centre for Quantum Software and Information (QSI), Faculty of Engineering and Information Technology, University of Technology Sydney, NSW 2007, Australia (e-mail: sanjiang.li@uts.edu.au).}       Xiangzhen~Zhou,\thanks{Xiangzhen Zhou is with Nanjing Tech University (email: xiangzhenzhou@njtech.edu.cn).}
Yuan~Feng \thanks{Yuan Feng is with the Department of Computer Science and Technology, Tsinghua University (email: yuan\_feng@tsinghua.edu.cn).}

}



\maketitle

\begin{abstract}
Current {superconducting} quantum devices impose strict connectivity constraints on quantum circuit execution, necessitating circuit transformation before executing quantum circuits on physical hardware. Numerous quantum circuit transformation (QCT) algorithms have been proposed.  To enable faithful evaluation of state-of-the-art QCT algorithms, this paper introduces \qknob\ (Qubit mapping Benchmark with Known Near-Optimality), a novel benchmark construction method for QCT. \qknob\ circuits have built-in transformations with near-optimal (close to the theoretical optimum) SWAP count and depth overhead. \qknob\ provides general and unbiased evaluation of QCT algorithms. Using \qknob, we demonstrate that \sabre, the default Qiskit compiler, consistently achieves the best performance on the 53-qubit IBM Q Rochester and Google Sycamore devices for both SWAP count and depth objectives. Our results also reveal significant performance gaps relative to the near-optimal transformation costs of \qknob. Our construction algorithm and benchmarks are open-source.
\end{abstract}


\section{Introduction}
\label{sec:intro}

{Superconducting} noisy intermediate-scale quantum (NISQ) devices impose strict connectivity constraints on quantum circuit execution. This necessitates a crucial compilation step known as \emph{quantum circuit transformation} (QCT, also referred to as transpilation, qubit mapping, or layout synthesis). QCT adapts ideal circuits for execution on physical quantum devices by ensuring that two-qubit gate (like CNOT or CZ) can only be performed between neighbouring qubits.

QCT has become a significant research focus in quantum computing \cite{ChildsSU19-qct,Cowtan+19-tket,Nannicini+21_bipmapping,Saeedi+11_synthesis,Venturelli+18_Planner}, electronic design automation \cite{Ash-Saki+19_qure,Deng0L20_codar,Itoko+19_commutation,TanC21-gate_absorption,Xie+21_commutativity,ZhouFL20_MCTS_iccad,Zhou+22_MCTS_Todaes,Zhou+20_SAHS,Zulehner+18_Astar}, and computer architecture \cite{Li+19-sabre,Liu+22_not_all_swap,Murali+19,TannuQ19, Zhang+21-time}. This paper refers to this procedure as, {largely} interchangeably, qubit mapping and (quantum) circuit transformation.  

Over the past several years, numerous QCT algorithms have been proposed for mapping ideal circuits to physical quantum devices \cite{ChildsSU19-qct,Cowtan+19-tket,Nannicini+21_bipmapping,Saeedi+11_synthesis,Venturelli+18_Planner,Ash-Saki+19_qure,Deng0L20_codar,Itoko+19_commutation,TanC21-gate_absorption,Xie+21_commutativity,ZhouFL20_MCTS_iccad,Zhou+22_MCTS_Todaes,Zhou+20_SAHS,Zulehner+18_Astar,Li+19-sabre,Liu+22_not_all_swap,Murali+19,TannuQ19, Zhang+21-time,Siraichi+18, Booth+18_Planning,Almeida+19_permutation,siraichi+19_bmt,LiZF21_fidls,li-iccad23sqgm,Niemann+21_combine_remote_cnot,Bandic23qmi,Huang24dasatom}. These algorithms typically involve constructing an initial mapping and inserting SWAP gates as needed to ensure that all two-qubit gates  comply with the device's connectivity constraints. The cost of a transformation is usually measured by the \emph{SWAP count} (number of inserted SWAP gates) or the \emph{depth overhead} (increase in circuit depth). Determining whether a transformation exists with a SWAP count or depth overhead below a given threshold has been proven to be NP-complete \cite{Siraichi+18,TanC21-queko}. 
Consequently, exact algorithms \cite{Nannicini+21_bipmapping, Venturelli+18_Planner, Booth+18_Planning,TanC20_iccad_optimal} are often computationally intractable for circuits with more than approximately 10 qubits. Therefore, most QCT algorithms are heuristic.

{Evaluating the performance of QCT algorithms is challenging, as their effectiveness depends on the target hardware architecture, input circuit structure, and optimisation objectives. Benchmarking offers a systematic approach to comparing QCT algorithms through standardised circuits and evaluation metrics \cite{li-tqc23qasmbench,quetschlich2023mqtbench,chen2022veriqbench,nation2024benchmarking,TanC21-queko}. These benchmarks enable controlled performance assessments, ensuring fair comparisons and reproducible results, while also facilitating the identification of algorithmic strengths and weaknesses to guide future development.}

{While reversible circuit benchmarks like RevLib (\url{https://www.revlib.org/}) and quantum circuits from QASMBench \cite{li-tqc23qasmbench} and MQTBench \cite{quetschlich2023mqtbench} are valuable for assessing scalability and practical applicability, they lack \emph{known} optimal transformation costs, hindering the evaluation of how closely a given transformation approximates the theoretical optimum. To address this issue, the \queko\ benchmark  \cite{TanC21-queko} was introduced, featuring circuits with zero-cost optimal transformations. However, \queko\ circuits may not be generalisable to typical QCT scenarios, potentially leading to skewed evaluation, particularly for algorithms relying on subgraph isomorphism for initial mappings (e.g., \cite{siraichi+19_bmt, LiZF21_fidls}). These limitations underscore the need for more versatile and representative benchmarking frameworks.}

\IEEEpubidadjcol

{To address these challenges, this paper introduces \qknob\ (\emph{Qubit mapping Benchmark with Known Near-Optimality}), a novel framework for evaluating QCT algorithms.  The construction of \qknob\ circuits is based on theoretical guarantees provided by Theorems~\ref{thm:qct} and~\ref{thm:quekno}. Theorem~\ref{thm:qct} demonstrates that any circuit transformation can be represented as a `partition-and-permute' process, where the circuit is partitioned into subcircuits and transformed via an initial mapping followed by a sequence of permutations. The transformation cost is determined by the number of SWAP gates required to implement these permutations. \qknob\ circuits are constructed by reversing this process: starting with subcircuits derived from subgraph-guided selections, we link them with restricted permutations. This construction, guaranteed by Theorem~\ref{thm:quekno}, ensures that each \qknob\ circuit has a built-in transformation cost determined by the number of SWAP gates required to realise the permutations. To ensure near-optimality, the construction method restricts the types of permutations used, such as those achievable with up to two consecutive SWAPs (for SWAP optimality) or parallel SWAP operations (for depth optimality). This restriction enables the generation of benchmarking circuits that are both representative of realistic quantum computing scenarios and adaptable to the connectivity constraints of various quantum devices.}

For SWAP count and depth optimality evaluation, we generated \qknob\ circuits for three representative quantum devices: IBM Q Tokyo (20 qubits), IBM Q Rochester (53 qubits), and Google Sycamore (53 qubits). {The construction process carefully selects interaction graphs of subcircuits, subgraph embeddings, and permutations to match the unique connectivity constraints of these devices, as will be explained shortly.} Using \qknob\ alongside \queko\ circuits, we evaluated five state-of-the-art QCT algorithms. Among these algorithms, \sabre\ \cite{Li+19-sabre}, the default Qiskit compiler, consistently demonstrated the best performance across both objectives and all three devices. Unlike \queko, \qknob\ benchmarks provided more faithful and unbiased evaluation. {In addition, our evaluation revealed significant gaps between the built-in transformation costs of \qknob\ circuits and the performance of even the best QCT algorithms.}

\vspace*{3mm}

The remainder of this paper is organised as follows. Sec.~\ref{sec:prelimaries} introduces the necessary background on quantum circuits, graphs, and permutations, while Sec.~\ref{sec:qct} discusses QCT process and algorithms. The theoretical foundations and design principles behind \qknob\ are presented in Sections~\ref{sec:quekno-theory} and~\ref{sec:design}. In Sec.~\ref{sec:evaluation}, we evaluate state-of-the-art QCT algorithms on both \qknob\ and \queko\ benchmark sets, providing comparative insights.  Sec.~\ref{sec:discussion} addresses \qknob's scalability and limitations, and Sec.~\ref{sec:conclusion} concludes the work. Proofs for the lemmas and theorems are included in the appendix.

\section{Preliminaries}\label{sec:prelimaries}

{This section covers relevant background on quantum circuits, subgraph isomorphism, and permutations, essential for describing and formalising the construction methods used in this work.}

\subsection{Quantum Circuits}
Quantum circuits are the standard model for describing quantum algorithms. Like classical combinational circuits, a quantum circuit consists of a sequence of quantum gates acting on qubits (quantum bits).  A general state of qubit $q$ has the form $\ket{\psi}=\alpha_0\ket{0}+\alpha_1\ket{1}$ with $\alpha_0,\alpha_1$ being complex values and $|\alpha_0|^2+|\alpha_1|^2=1$.  

Quantum gates are unitary transformations. An $n$-qubit gate is represented as a $2^n\times 2^n$ unitary matrix.  The following one-qubit gates are often used:
\begin{equation*}
\resizebox{0.9\hsize}{!}{%
$X = \begin{pmatrix}
    0& 1\\ 1 & 0
    \end{pmatrix}$,\quad 
    $H = \frac{1}{\sqrt{2}}\begin{pmatrix}
    1& 1\\ 1 & -1
    \end{pmatrix}$,\quad  
    $S = \begin{pmatrix}
    1& 0\\ 0 & i
    \end{pmatrix}$, \quad 
    $T = \begin{pmatrix}
    1& 0\\ 0 & e^{i\frac{\pi}{4}}
    \end{pmatrix}.$
    }
\end{equation*}
CNOT (also called CX) and CZ are two-qubit gates. For any computational basis state $\ket{i}\ket{j}$, CNOT and CZ map $\ket{i}\ket{j}$ to, respectively, $\ket{i}\ket{i\oplus j}$ and $(-1)^{i\cdot j}\ket{i}\ket{j}$, where $\oplus$ denotes exclusive-or and $\cdot$ denotes Boolean conjunction. 

One-qubit gates and CNOT are sufficient to implement an arbitrary quantum gate. In addition, any quantum gate can be approximated with arbitrary precision using only $H, S, T$ and CNOT gates. In particular, the two-qubit SWAP gate, which maps $\ket{i}\ket{j}$ to $\ket{j}\ket{i}$, can be implemented by three CNOT gates, i.e., $\swap{p}{q} = \textsc{cnot}(p,q)\textsc{cnot}(q,p)\textsc{cnot}(p,q)$. 

While different quantum devices may support different universal sets of quantum gates, in practice, the two-qubit gate in such a universal set is CNOT or CZ. As the actual functionality of a one-qubit gate plays no role in QCT, in the rest of this paper, we denote a one-qubit gate simply by the qubit it acts on. For example, an $H$ gate on $q_i$ is simply denoted by $\la q_i\ra$. Analogously, we write a CNOT or CZ gate with control qubit  $q_i$ and target qubit $q_j$ simply as $\la q_i,q_j\ra$. 

\begin{figure}
\begin{tabular}{ll}
\begin{tabular}{l}
\scalebox{0.7}{
\centerline{
\Qcircuit @C=0.8em @R=1.1em {
\lstick{q_{0}}  & \qw  &\targ  \barrier{5}&\qw \barrier{5}&\qw & \targ \barrier{5}   &\gate{H} \barrier{5} & \gate{T} \barrier{5}& \targ  &\qw \barrier{5}  & \qw 
       & \qw \barrier{5}  & \qw\barrier{5} & \qw &\qw\\
\lstick{q_{1}}  & \targ  &\qw & \gate{H}  & \gate{T} &\qw &\gate{H}    & \ctrl{3}& \qw & \gate{T} 
      &\gate{H} &\qw  & \qw & \qw &\qw  \\
\lstick{q_{2}} &\ctrl{-1}   & \qw & \gate{H} & \qw   &\qw  &\qw & \qw & \qw 
      & \targ &\qw & \ctrl{2}  &\gate{T}    &\gate{H}& \qw   \\
\lstick{q_{3}} & \gate{H} & \qw &\qw & \targ &\qw &\qw    & \qw     &  \ctrl{-3} & \qw
        & \gate{T} & \qw  & \targ &\qw & \qw \\
 \lstick{q_{4}} & \qw  & \ctrl{-4} & \gate{H}&\qw &  \ctrl{-4} &\qw  & \targ& \qw
       & \ctrl{-2} & \qw & \targ & \qw  & \qw & \qw  \\
 \lstick{q_{5}} & \gate{H} & \qw  &\gate{T} & \ctrl{-2} & \qw & \gate{H}&\qw &\qw
         &  \qw & \qw  &\qw   & \ctrl{-2}& \qw & \qw \\
  }
  }
  }
\end{tabular}
& \hspace{-5mm}
\begin{tabular}{l}
\includegraphics[width=0.08\textwidth]{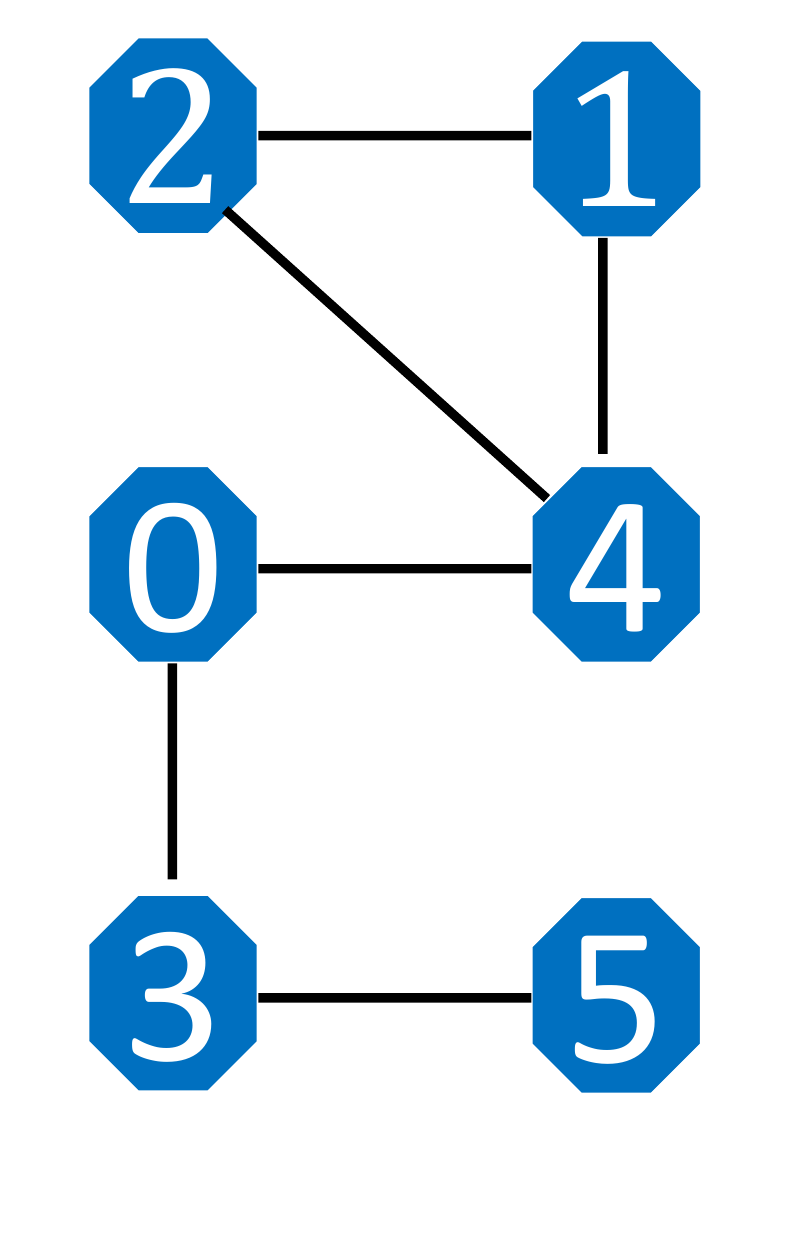}
\end{tabular}
\end{tabular}
 \caption{A quantum circuit (left) and its interaction graph (right). }
\label{fig:quantum_circuit}
\end{figure}

A circuit $C$ is usually represented as a sequence of gates $g_0,g_1,\ldots,g_{M-1}$, but this sequence does not necessarily reflect execution order. Gates acting on distinct qubits may be executed in parallel. Naturally, we partition $C$ into layers, scheduling each gate as early as possible. The \emph{depth} of a circuit is the number of its layers. 

For example, the circuit in Fig.~\ref{fig:quantum_circuit} has a depth of 9, with its first layer containing gates $\la 5\ra, \la 2, 1\ra,\la 4, 0\ra,\la 3\ra $.

\subsection{{Graphs and} Subgraph Isomorphism}
{Graphs naturally arise in quantum circuit transformation.} A quantum device is represented as an undirected \textit{architecture graph} $G = (V, E)$, where vertices in $V$ are physical qubits and edges in $E$ represent allowed two-qubit interactions. That is,  $(p,q)\in E$ if and only if a two-qubit gate can be applied to qubits $p$ and $q$. Since $G$ is undirected, $(p,q)\in E$ if and only if $(q,p)\in E$. Fig.~\ref{fig:ags} shows the architecture graphs of $Grid(3,2)$ (an artificial device) and IBM Q Rochester. {Each quantum circuit also induces an \textit{interaction graph}.}  
\begin{definition}[Interaction Graph]
Let $C$ be a quantum circuit on qubit set $Q$. The interaction graph $(Q,E)$ of $C$ is defined as: for all $p,q\in Q$, $(p,q)\in E$ if and only if $\cx{p}{q}$ or $\cx{q}{p}$ is in $C$.
\end{definition}

{The well-known subgraph isomorphism problem plays a critical role in quantum circuit transformation (see Sec.~\ref{sec:qct}). Specifically, if a subgraph isomorphism $f$ exists that embeds the circuit's interaction graph into the device's architecture graph, the circuit can be executed directly using $f$, eliminating the need for SWAP insertions.}

\begin{definition} [Subgraph Isomorphism]\label{dfn:subgraph_isomorphism}
Given two graphs $G_i=(V_i,E_i)$ $(i=1,2$), we say $G_1$ is a subgraph of $G_2$ if $V_1\subseteq V_2$ and $E_1\subseteq E_2$. $G_1$ is \emph{embeddable} into $G_2$ if there is an injective mapping $f:V_1\to V_2$ such that $(f(p),f(q))\in E_2$ for any edge $(p,q)\in E_1$. In this case, $f$ is called a \emph{subgraph isomorphism} or an \emph{embedding}. 
\end{definition}
Subgraph isomorphisms can be found or disproved by algorithms like VF2  \cite{Cordella+04-vf2}. A subgraph isomorphism can also be quickly disproved by identifying a property that $G_1$ has but $G_2$ lacks (e.g., the presence of a 3-cycle). A 3-cycle in a graph $G$ is a path $(v_0, v_1,v_2,v_3)$ of length 3 in $G$ such that $v_0=v_3$ while $v_0,v_1,v_2$ are pairwise different.  

Consider the circuit $C$  in Fig.~\ref{fig:quantum_circuit} (left).  Its interaction graph contains a 3-cycle $(2,1,4,2)$ (see  Fig.~\ref{fig:quantum_circuit} (right)). Because $Grid(3,2)$ has no 3-cycles, the interaction graph cannot be embedded in $Grid(3,2)$. 

\begin{figure}[tbp]
    \centering
    \begin{tabular}{cc}
    \includegraphics[width=0.1\textwidth]{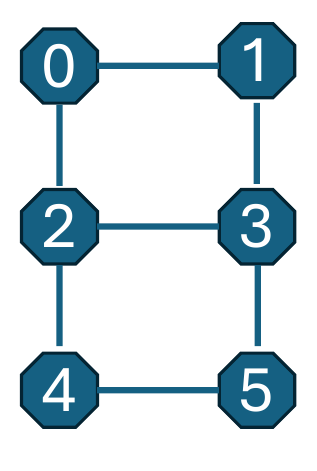} &{\hspace*{6mm}} \includegraphics[width=0.15\textwidth]{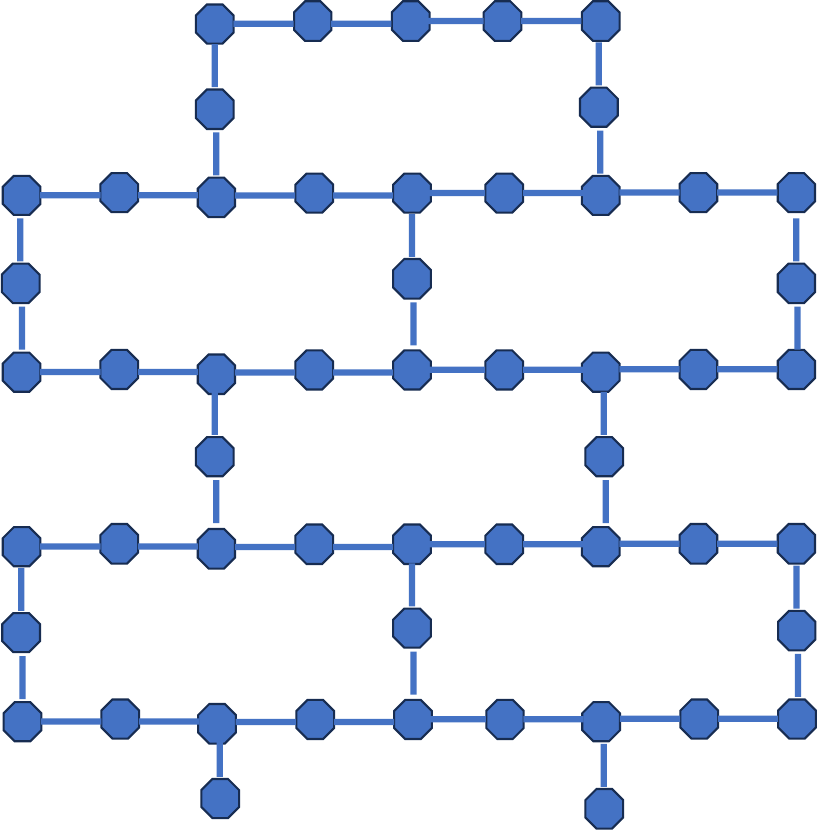}
    \end{tabular}
    \caption{The architecture graphs of $Grid(3,2)$ (left) and IBM Q Rochester (right).}
    \label{fig:ags}
\end{figure}

\subsection{Permutations and SWAP Circuits}
In this paper, we refer to the circuit before transformation as a \emph{logical circuit} and the transformed circuit as a \emph{physical circuit}. Similarly, qubits in a logical (physical) circuit are referred to as logical (physical) qubits. Note that the term `logical' used in QCT should not be confused with its use in error correction.

{Our benchmark circuit construction relies heavily on SWAP circuits (i.e., circuits consisting of SWAP gates) and permutations of physical qubits on the architecture graph. Permutations, implemented via SWAP circuits, modify the mappings from logical to physical qubits; this modification through permutation is essential for constructing circuits with known transformation costs.}

Let $G=(V,E)$ be an undirected graph. Assume $V= [n]\define\set{0,1,\ldots,n-1}$. A \textit{permutation} $\pi: V \to V$ is a bijection on $V$. For example, the identity mapping $id_V$ is a permutation; a SWAP operation on an edge $(i,j)$ of $G$ induces a permutation $\pi_{i,j}$ which maps $i$ to $j$ and $j$ to $i$, leaving other vertices unchanged. Permutations can be composed to form more complicated permutations. 

Formally, a permutation $\pi$ is implementable by a sequence of $c$ SWAPs $\pi_{p_1,q_1},\ldots,\pi_{p_c,q_c}$ if {$\pi=\pi_{p_c,q_c}\circ\cdots\circ \pi_{p_1,q_1}$}, where $(p_i,q_i)$ is an edge in $G$ for $1\leq i\leq c$. Any permutation $\pi$ on a connected graph $G$ can be implemented by a sequence of SWAPs on edges of $G$. We denote by $\swapnorm{\pi}$ the minimum number of SWAPs required to implement $\pi$ and call this the \emph{swap cost} of (implementing) $\pi$.

Since $V=[n]$, we represent each permutation $\pi$ on $V$ as the vector $(\pi(0),\pi(1),\ldots,\pi(n-1))$.  
For example,  $\pi=(3,0,2,1,4,5)$ denotes the permutation on {$V=[6]$} that maps 0 to 3, 1 to 0, 3 to 1 and leaves the other vertices unchanged. We can implement $\pi$ by first swapping 0 and 1 and then swapping 1 and 3. That is, {$\pi = \pi_{1,3}\circ \pi_{0,1}$.} 
Note that permutation composition is not commutative. For example,   ${\pi_{0,1}\circ \pi_{1,3}}=(1,3,2,0,4,5)\not=\pi$. 

Regarding quantum circuit transformation, for each edge $(i,j)$ in an architecture graph $G=(V,E)$, the permutation $\pi_{i,j}$ corresponds to the SWAP gate $\swap{i}{j}$. 

{Next, we explain how a permutation on $G$ can be used to modify a logical-to-physical qubit mapping using SWAP gates.} Suppose $C$ is a quantum circuit on qubit set $Q$. For convenience, we assume $Q\subseteq V$. Let $\sigma:Q\to V$ be the current (logical to physical) mapping. Applying a SWAP gate $\swap{i}{j}$ transforms $\sigma$ into a new mapping $\sigma'=\pi_{i,j}\circ\sigma$, where $\sigma'(p)=j$ if $\sigma(p)=i$, $\sigma'(p)=i$ if $\sigma(p)=j$, and $\sigma'(p)=\sigma(p)$ otherwise. In general, for a SWAP circuit  $S \define (\swap{p_1}{q_1}, \ldots,\swap{p_c}{q_c})$, $S$ transforms $\sigma$ into $\sigma'\define\pi_{p_c,q_c}\circ\cdots\circ\pi_{p_1,q_1} \circ \sigma$. In this case, we say that $\pi \define \pi_{p_c,q_c}\circ\cdots\circ\pi_{p_1,q_1}$ is implemented by $S$. 
Therefore, any permutation $\pi$ on $G$ can be implemented by a SWAP circuit to modify a logical-to-physical mapping. We denote by $\sem{\pi}$ a SWAP circuit that implements $\pi$ with a minimal number of SWAP gates.

{The following lemma formalises how the concatenation of SWAP circuits implements the composition of their corresponding permutations on the architecture graph.}
\begin{lemma}\label{lem:permcomp}
Let $\pi_1,\pi_2$ be two permutations on a graph $G=(V,E)$. Suppose $S_i$ is a SWAP circuit that implements $\pi_i$ for $i=1,2$. Then $S_1+S_2$ implements {$\pi_2\circ\pi_1$}, where `+' denotes circuit concatenation.
\end{lemma}

When constructing \qknob\ circuits, we often permute subgraphs and circuits. {These procedures formalises how subgraphs and circuits are rearranged during circuit construction, ensuring that the relationships between qubits and connectivity constraints are accurately represented.}

Given a graph $G$, a permuted graph is obtained by rearranging its nodes.

\begin{definition}[Permuted Subgraph]\label{dfn:pgraph}
Let $G=(V,E)$ be an undirected graph and  $G_1=(V_1,E_1)$ a subgraph of $G$. The permutation of $G_1$ under a permutation $\pi$ on $V$, denoted $\pi(G_1)$, is the graph  $(\pi(V_1),\pi(E_1))$, where $\pi(V_1) \define \{\pi(v) \mid v\in V_1\}$ and $\pi(E_1) \define \{ (\pi(v),\pi(v')) \mid (v,v')\in E_1\}$. 
\end{definition}

Analogously, a permuted circuit is obtained by rearranging its qubits.
\begin{definition}[Permuted Circuit]\label{dfn:pcirc}
For a  circuit $C\define(g_1,\ldots,g_M)$ on qubits in $V$,  the permutation of $C$ under a permutation $\pi$ on $V$ is $\pi(C) \define (\pi(g_1),\ldots,\pi(g_M))$, where $\pi(g)$ is the same gate as $g$ but operates on 
\begin{itemize}
    \item qubit $\pi(q_i)$ if $g$ is a one-qubit gate on $q_i$,
    \item qubits $\pi(q_i)$ and $\pi(q_j)$ if $g$ is a two-qubit gate on $q_i,q_j$. 
\end{itemize} 
\end{definition}

{The following lemma establishes the reversibility of permutations on circuits and their distributivity over circuit concatenation, which are key properties used in the construction and analysis of \qknob\ circuits.}

\begin{lemma}\label{lem:2}
For a permutation $\pi$ and circuits $C,C_1,C_2$, we have $\pi^{-1}(\pi(C))=C$ and $\pi(C_1+C_2) = \pi(C_1) + \pi(C_2)$. 
\end{lemma}

\section{Quantum Circuit Transformation}
\label{sec:qct}
{Quantum circuit transformation (QCT) is crucial for compiling quantum algorithms to execute on quantum devices with limited connectivity. This section outlines the fundamental steps involved in QCT, including mapping transformations and gate executions, and provides a theoretical framework for understanding the cost and structure of these transformations, focusing on core QCT components and their relevance to \qknob\ circuit construction.}

\subsection{Overview of Quantum Circuit Transformation}\label{sec:qct_overview}

When designing a quantum algorithm using the quantum circuit model, the designer often does not consider a specific quantum device. Consequently, the resulting ideal quantum circuit may contain two-qubit gates acting on arbitrary qubit pairs, which may violate hardware connectivity constraints. Therefore, QCT is necessary.

Let $\ag \define (\V, \E)$ be the architecture graph of a given quantum device.
Device-supported one-qubit gates can be executed directly on $\ag$.  A device-supported two-qubit gate is \emph{directly executable} on $\ag$ if its two qubits are adjacent in $\ag$. A quantum circuit $C$ is referred to as an $\ag$-circuit if all its two-qubit gates are directly executable given the device's connectivity constraints. 

If $C$ is not an $\ag$-circuit, transformation is necessary for execution. Let $Q$ be the qubit set of $C$, assuming w.l.o.g. that $Q\subseteq \V$.  To execute $C$ on $\ag$, we first construct an \textit{initial mapping} $\sigma_1: Q \to \V$. If all gates in $C$ are executable under $\sigma_1$ (i.e., for every two-qubit gate $\cx{u}{v}$ in $C$, $(\sigma_1(u),\sigma_1(v))\in \E$), we say that $C$ is executable on $\ag$ under $\sigma_1$. This occurs if and only if $\sigma_1(C)$ is an $\ag$-circuit. 

If $C$ is not executable under $\sigma_1$, QCT alternates between two key procedures:
\begin{enumerate}
    \item \textbf{Mapping Transformation}: Modify the current mapping by inserting SWAP gates.
    \item \textbf{Gate Execution}: Remove gates from the logical circuit that become executable under the updated mapping.
\end{enumerate}
The process repeats until all gates in $C$ have been removed.

\subsection{Transformation by Partitioning and Permuting} \label{sec:part}
A straightforward, though not necessarily optimal, approach to QCT (see \cite{siraichi+19_bmt,LiZF21_fidls}) involves partitioning the input circuit $C$ into nonempty subcircuits ${C}_1,\ldots,{C}_s$,  such that each  ${C}_i$ is executable under a specific injective mapping  $\sigma_i: Q_i\to \V$, where $Q_i\subseteq Q$ is the set of qubits in $C_i$. Each transformed subcircuit, denoted $\widetilde{C}_i\define \sigma_i(C_i)$, becomes  an $\ag$-circuit. Since $C_i = \sigma^{-1}_i(\widetilde{C}_i)$ for $1\leq i\leq s$, we have
\begin{align}\label{eq:circpart}
C = C_1 +\cdots + C_s=\sigma_1^{-1}(\widetilde{C}_1) +  \cdots + \sigma_s^{-1}(\widetilde{C}_s),
\end{align}
where each $\sigma_i^{-1}$ is defined on $\sigma_i(Q)\subseteq \V$. 

The transformation and execution proceed as follows:

\vspace*{1mm}
\noindent\textbf{Initial Mapping Application:} Apply $\sigma_1$ to $C$. This transforms in particular the first subcircuit into the $\ag$-circuit $\widetilde{C}_1$. We remove all gates in $C_1$ from $C$ (since they are all executable), setting $PC_1 = \widetilde{C}_1$ for the current physical circuit and $LC_1 = C_2+\cdots + C_s$ for the remaining logical circuit. We then append  to $PC_{1}$ the SWAP circuit $\sem{\sigma_1^{-1}} $, which implements ${\sigma_1^{-1}}$ with the minimal number of SWAPs. This does not change $LC_1$ but updates  $PC_1$ to $\widetilde{C}_1+\sem  {\sigma_1^{-1}} $ and reverts the current mapping ${\sigma_1}$ to the identity mapping $id_{\V}$. 
    
\noindent\textbf{Subsequent Subcircuits:} Suppose the current mapping is $id_{\V}$ and the logical and physical circuits are $LC_{i-1}=C_i+\cdots + C_s$ and $PC_{i-1}$ for some $1<i<s$. We append a SWAP circuit $\sem{\sigma_i} $ to $PC_{i-1}$ immediately after $\sem{\sigma_{i-1}^{-1}} $. This changes the current mapping to $\sigma_i$ and transforms the $i$-th subcircuit $C_i$ to $\widetilde{C}_i$. We append $\widetilde{C}_i$ to $PC_{i-1}$, obtaining 
    \[
    \resizebox{1\hsize}{!}{%
    $PC_i = {\widetilde{C}_1+\sem{\sigma_1^{-1}}  +\sem{\sigma_2} + \cdots + \widetilde{C}_{i-1} + \sem{\sigma_{i-1}^{-1}} + \sem{\sigma_i} + \widetilde{C}_i.}$
    }
    \] 
The logical circuit is updated to $LC_i = C_{i+1}+\cdots+C_s$. Next, we append a SWAP circuit $\sem{\sigma_i^{-1}} $ to $PC_i$. This does not change $LC_i$ but updates $PC_i$ to $PC_i+\sem{\sigma_i^{-1}} $ and reverts the current mapping $\sigma_{i}$ to $id_{\V}$.

\noindent\textbf{The Last Subcircuit:} In the final step ($i=s-1$), $LC_i = C_s$, and the current mapping is $id_{\V}$. We append a SWAP circuit $\sem{\sigma_s}$ to $PC_{s-1}$. This changes the current mapping to ${\sigma_s}$ and transforms the last subcircuit $C_s$ to $\widetilde{C}_s$. We remove all gates in $C_s$ from $LC_{s-1}$, resulting an empty logical circuit. Accordingly, the final physical circuit is 
    \begin{align*}
    \resizebox{1\hsize}{!}{%
           {$PC_s = \widetilde{C}_1 + \sem{\sigma_1^{-1}} +\sem{\sigma_2} + \cdots +   \widetilde{C}_{s-1} + \sem{\sigma_{s-1}^{-1}} + \sem{\sigma_s}+\widetilde{C}_s$.}
           }
    \end{align*}

By Lemma~\ref{lem:permcomp}, two consecutive SWAP circuits can be replaced with a single circuit. In particular, we can replace {$\sem{\sigma_i^{-1} } + \sem{\sigma_{i+1}}$ with $\sem{  \sigma_{i+1}\circ\sigma_i^{-1}}$}. Thus, the final physical circuit can be written as 
\begin{align}\label{eq:physical-circuit-final}
\resizebox{0.9\hsize}{!}{%
    {$PC = \widetilde{C}_1 + \sem{\sigma_2\circ \sigma_1^{-1}}+ \cdots + \widetilde{C}_{s-1} + \sem{\sigma_s\circ\sigma_{s-1}^{-1}}+\widetilde{C}_s$.}
    }
\end{align}
The transformation ensures equivalence between $PC$ and $C$, up to an initial mapping $\sigma_1$ and a final mapping $\sigma_s$.

In summary, the logical circuit $C$ has been transformed into the physical circuit $PC$ by (i) applying the initial mapping $\sigma_1$ on $C$, and (ii) inserting SWAP circuit $S_i\define \sem{\sigma_{i+1}\circ\sigma_{i}^{-1}}$ between $C_i$ and $C_{i+1}$ for $1\leq i\leq s-1$. As the initial and final mappings incur no cost, the total transformation cost is the number of SWAPs used in the SWAP circuits $\sem{\sigma_{i+1}\circ\sigma_i^{-1} }$ for $1\leq i<s$, i.e., the number of SWAPs  in $S_1+\cdots +S_{s-1}$.

\vspace*{2mm}
\noindent\textbf{A Key Theoretical Result:} The above process provides a specific circuit transformation that can be represented as a sequence of subcircuits and permutations. Moreover, any arbitrary transformation of $C$ into a physical circuit $PC$ on $\ag$ can be achieved in the same manner. Specifically, let $\sigma_1$ be the initial mapping, and let $S_1,\ldots,S_{s-1}$ be the sequence of SWAP circuits inserted into $C$. We can partition the circuit $C$ into subcircuits $C_1,\ldots,C_s$ as in Eq.~\ref{eq:circpart} and represent $PC$ as in Eq.~\ref{eq:physical-circuit-final}, where $\sigma_i$ $(1\leq i\leq s)$ are permutations such that $S_i=\sem  {\sigma_{i+1}\circ\sigma_i^{-1} }$.

This leads to the following theorem:
\begin{theorem}\label{thm:qct}
Let $C$ be a logical circuit and $\ag$ the architecture graph of a quantum device. For any transformation of $C$ on $\ag$ with cost $c$, there exist a partition of $C$ into $1\leq s\leq c+1$ nonempty subcircuits $C_1,\ldots,C_s$ and $s$ permutations $\sigma_i$ $(1\leq i\leq s)$ on $\ag$ such that $\sum_{i=2}^s \swapnorm{{\sigma_{i}\circ\sigma_{i-1}^{-1}}} = c$, and for each $1\leq i\leq s$,
$\widetilde{C}_i\define \sigma_i(C_i)$ is an $\ag$-circuit. Moreover, the transformed circuit has the form shown in Eq.~\ref{eq:physical-circuit-final}.
\end{theorem}
{The circuit transformation form described in Theorem~\ref{thm:qct} and shown in Eq.~\ref{eq:physical-circuit-final} is called the partition-and-permute transformation.}

\subsection{A Circuit Transformation Example}\label{sec:qct-ex}
{This subsection provides an example demonstrating how initial mapping and SWAP insertions resolve connectivity constraints and how the transformation can be expressed in the partition-and-permute form described in Theorem~\ref{thm:qct}.}

Let $\ag$ be the $Grid(3,2)$ architecture graph (see Fig.~\ref{fig:ags}, left); and consider the logical circuit shown in Fig.~\ref{fig:quantum_circuit}. Suppose our target is to minimise the SWAP count. Because one-qubit gates can be executed directly, we remove them, leaving the logical circuit
\begin{align*}
\resizebox{1\hsize}{!}{%
     $C =   \big[\la 2, 1\ra, \la 4, 0\ra, \la 4, 0\ra, \la 5, 3\ra, \la 1, 4\ra, \la 4, 2\ra, \la 2, 4\ra, \la 3, 0\ra, \la 5, 3\ra\big].$
     }
\end{align*}
Because $C$'s interaction graph (Fig.~\ref{fig:quantum_circuit}) contains a 3-cycle $(2,1,4,2)$, it cannot be embedded into $\ag$; thus this circuit is not executable under any initial mapping. We next show that it can be transformed into an executable circuit with one SWAP. 

First, we partition $C$ into two subcircuits such that each subcircuit's interaction graph is embeddable in $\ag$. For example, let 
$C_1 = [\la 2, 1\ra, \la 4, 0\ra, \la 4, 0\ra, \la 5, 3\ra, \la 1, 4\ra]$, and
$C_2 = [ \la 4, 2\ra, \la 2, 4\ra, \la 3, 0\ra, \la 5, 3\ra]$. Using a subgraph isomorphism algorithm, we find a mapping  $\sigma_1= (5,1,0,4,3,2)$ that transforms $C_1$ into an $\ag$-circuit. 

Note that $\sigma_1$ permutes the whole circuit  as
\begin{align*}
\resizebox{1\hsize}{!}{%
$ \sigma_1(C) =   
\big[\la 0, 1\ra, \la 3, 5\ra, \la 3, 5\ra, \la 2, 4\ra, \la 1, 3\ra, \la 3, 0\ra, \la 0, 3\ra, \la 4, 5\ra, \la 2, 4\ra \big]$.
}
\end{align*}
Because all gates in $\sigma_1(C_1)$ act on neighbouring physical qubits in $\ag$, they are executable and therefore removed from $\sigma_1(C)$. The remaining circuit is $\sigma_1(C_2) = [\la 3, 0\ra, \la 0, 3\ra$, $\la 4, 5\ra, \la 2, 4\ra]$. Because $\la 2,4\ra$ and $\la 4,5\ra$ correspond to edges in $\ag$, we only need bring the logical qubit mapped to physical qubit $0$ next to the logical qubit mapped to physical qubit $3$ (or vice versa). This can be achieved by inserting the SWAP gate $\swap{0}{1}$, which implements the permutation $\pi_{0,1}$. Note that $\pi_{0,1}^{-1}=\pi_{0,1}$. Because $\pi_{0,1}\big(\sigma_1(C_2)\big)= [\la 3,1\ra,\la 1,3\ra,\la 4,5\ra,\la 2,4\ra]$ is an $\ag$-circuit, it can be executed directly. 

In conclusion, to transform $C$, we apply an initial mapping $\sigma_1$ and then insert $\swap{0}{1}$. Because the SWAP cost is 1, the transformation is optimal. 

Let $\sigma_2 = \pi_{0,1} \circ \sigma_1$. Then, the circuit $C_i$ is executable under $\sigma_i$ for $i = 1, 2$. By definition, $\sigma_{2} \circ \sigma_1^{-1} = (\pi_{0,1} \circ \sigma_1) \circ \sigma_1^{-1} = \pi_{0,1}$. Consequently, the set $S_1 = \sem{ \sigma_{2} \circ \sigma_1^{-1} }$ consists of a single $\swap{0}{1}$. This construction explicitly relates the above transformation to the form described in Theorem~\ref{thm:qct}.

\subsection{{Circuit Transformation Implementations}}

Many QCT algorithms have been developed. Some aim to find transformations with the optimal cost \cite{Saeedi+11_synthesis, Siraichi+18, Venturelli+18_Planner, Booth+18_Planning,Murali+19,Almeida+19_permutation,Nannicini+21_bipmapping, TanC20_iccad_optimal}, but these are applicable only to circuits with approximately ten or fewer qubits. Most algorithms employ heuristic search techniques \cite{Zulehner+18_Astar,Li+19-sabre,ChildsSU19-qct,Cowtan+19-tket,Zhou+20_SAHS,siraichi+19_bmt,LiZF21_fidls,ZhouFL20_MCTS_iccad}. 
These heuristic algorithms can be further classified by their optimisation objective. Some aim to maximise fidelity or minimise the error rate \cite{Murali+19,Ash-Saki+19_qure,TannuQ19,Deng0L20_codar}; many aim to reduce the SWAP count \cite{Zulehner+18_Astar,Li+19-sabre,LiZF21_fidls,Zhou+20_SAHS,ZhouFL20_MCTS_iccad}; and others aim to reduce the output circuit's depth \cite{Cowtan+19-tket,Zhang+21-time,Zhou+22_MCTS_Todaes,li-iccad23sqgm} to respect the limited qubit coherence time.  

Several QCT algorithms have used subgraph isomorphism. The BMT algorithm \cite{siraichi+19_bmt} combines subgraph isomorphism with token swapping and is closely related to the general transformation pattern as specified in Eq.~\ref{eq:circpart}. Based on exhaustive search,  BMT consumes significant time and memory, making it unsuitable for circuits with 20 or more qubits. Using the subgraph isomorphism algorithm VF2 \cite{Cordella+04-vf2}, \fidls\ \cite{LiZF21_fidls} selects an initial mapping that brings the input circuit's interaction graph closer to the architecture graph. When circuits are directly executable, \fidls\ can often find an embedding and transform the circuit without inserting any SWAPs. This is especially true for \queko\  circuits \cite{TanC21-queko}. 
A recent work \cite{Huang24dasatom} also employs subgraph isomorphism in a divide-and-shuttle atom approach for qubit mapping on neutral-atom quantum devices.

Other works exploit commutation rules \cite{Itoko+19_commutation,Xie+21_commutativity}, synthesis \cite{TanC21-gate_absorption}, gate optimisations \cite{Liu+22_not_all_swap}, or remote CNOT \cite{Niemann+21_combine_remote_cnot} for circuit transformation. These techniques can be combined with most QCT algorithms.  In this paper, we do not consider these extensions.

Our evaluation focuses on five state-of-the-art QCT algorithms: \tket, \sabre, \sahs, \mcts, and the Qiskit transpiler \texttt{StochasticSWAP}.

\tket\ is a powerful quantum circuit compiler tool proposed by Quantinuum. First described in \cite{Cowtan+19-tket}, the \tket\ router attempts to select the SWAP operation which maximally reduces the interaction graph's diameter for the current layer. 
\sabre\ \cite{Li+19-sabre} adopts a three-fold strategy and a heuristic function that models the fitness of a mapping with the two-qubit gates in the first several layers of the circuit. Specifically, \sabre\ starts from a random initial mapping and transforms the input circuit $C$ using the heuristic function; it then uses the final mapping as the initial mapping and transforms the reverse of $C$; and finally, it uses the final mapping of the second transformation as the initial mapping and transforms $C$ again. This approach incorporates information about the entire circuit. The \sahs\ algorithm \cite{Zhou+20_SAHS} first selects an initial mapping that best fits the input circuit $C$ using the simulated annealing method and then, during routing, simulates the search process one step further, selecting the SWAP with the best subsequent SWAP to apply. The Monte Carlo Tree Search (MCTS) method for QCT, denoted \mcts, was first introduced in  \cite{ZhouFL20_MCTS_iccad} for SWAP count optimisation and extended in \cite{Zhou+22_MCTS_Todaes} for depth optimisation. The core idea is to explore the search space in a balanced way. On average, MCTS-based QCT algorithms can search deeper and find better solutions.

\section{{\qknob} Circuits} \label{sec:quekno-theory}
{This section describes our method for constructing \qknob\ benchmarking circuits. These circuits are designed to address limitations in existing benchmarks \cite{TanC21-queko} by providing instances with known transformation costs and varying complexities.} Our procedure for constructing  {\qknob} circuits is  the \emph{reverse} of the partition-and-permute circuit transformation procedure described in Theorem~\ref{thm:qct} of Sec.~\ref{sec:qct}. 

\subsection{\qknob\ Circuits with Two Subcircuits}
To illustrate the construction, consider the task of creating a benchmarking circuit with an optimal SWAP cost 1. If a circuit's  interaction graph is not embeddable in $\ag$, its transformation cost is at least 1. We randomly generate subgraphs $G_1,G_2$ of $\ag$, and permute $G_2$ by the permutation $\pi_{i,j}$ for a randomly selected edge $(i,j)$ of $\ag$. 

We then generate $\ag$-circuits $\widetilde{C}_1$ and $\widetilde{C}_2$ with interaction graphs $G_1$ and $G_2$, respectively. Permuting $\widetilde{C}_2$ with $\pi_{i,j}$ yields a new circuit  $\widetilde{C}_1+ \pi_{i,j}(\widetilde{C}_2)$, with an optimal cost of at most 1. To increase the generality of the generated circuit, we apply another permutation $\pi_1$ to $\widetilde{C}_1+ \pi_{i,j}(\widetilde{C}_2)$. The resulting {circuit}, $\pi_1\big(\widetilde{C}_1+ \pi_{i,j}(\widetilde{C}_2)\big)$, is a {\qknob} circuit with an optimal SWAP cost of at most 1. Equality holds when the union of $G_1$ and $\pi_{i,j}(G_2)$ is not embeddable in $\ag$. 

{The general case involves the following notion of a subgraph link (glink), which connects the interaction graphs of two subcircuits.} 
\begin{definition}[Glink]\label{dfn:glink}
Let $\ag=(\V,\E)$ be an architecture graph and $G_i=(V_i,E_i)$ ($i=1,2$) its subgraphs. Suppose $\pi$ is a permutation on $\V$. Define the graph $G'=(V',E')$, where 
\begin{itemize}
    \item $V'=V_1\cup \pi(V_2)$, and
    \item $E' = E_1\cup \pi(E_2)$, that is,  for $u,v\in V'$, $(u,v)\in E'$ if and only if $(u,v)\in E_1$ or $(\pi^{-1}(u),\pi^{-1}(v))$ $\in E_2$.
\end{itemize}  
We call  $\la G_1,\pi,G_2\ra$ a \emph{subgraph link} (\emph{glink} for short), denoted by $\glink{G_1}{\pi}{G_2}$. When $G'$ is not embeddable in $\ag$, we call  $\la G_1,\pi,G_2 \ra$ a \emph{strong glink}, denoted by $\sglink{G_1}{\pi}{G_2}$.
\end{definition}
Fig.~\ref{fig:glink} shows an example of strong glinks.

Starting from a glink $\glink{G_1}{\pi_2}{G_2}$, we construct a circuit as follows: first, randomly generate $\ag$-circuits $\widetilde{C}_i$ for $i=1,2$ such that the $\ag$-subgraph $G_i$ is the interaction graph of $\widetilde{C}_i$; second, apply a randomly generated permutation $\pi_2$ to $\widetilde{C}_2$ and concatenate it with $\widetilde{C}_1$; finally, apply another permutation $\pi_1$ to  $\widetilde{C}_1 + \pi_2(\widetilde{C}_2)$. Let $C=\pi_1(\widetilde{C}_1)+(\pi_1\circ\pi_2)(\widetilde{C}_2)$. Then, $C$ is a \qknob\ circuit constructed from the glink $\glink{G_1}{\pi_2}{G_2}$. 

{We show that $C$ has a built-in transformation. The following lemma guarantees that, before applying $\pi_1$, the circuit can be transformed into an $\ag$-circuit by inserting any SWAP circuit that implements $\pi_2^{-1}$ between $\widetilde{C}_1$ and $\pi_2(\widetilde{C}_2)$.}
\begin{lemma}\label{lem:glink}
Let $\glink{G_1}{\pi_2}{G_2}$ be a glink. Suppose $\widetilde{C}_i$ $(i=1,2)$ is a circuit whose interaction graph is $G_i$. Then $\widetilde{C}_1 + \pi_2(\widetilde{C}_2)$ can be transformed into an $\ag$-executable circuit by inserting between $\widetilde{C}_1$ and $\pi_2(\widetilde{C}_2)$ any SWAP circuit that implements {$\pi_2^{-1}$}. 
\end{lemma}
{Furthermore, because $C$ is obtained by applying $\pi_1$ on $\widetilde{C}_1 + \pi_2(\widetilde{C}_2)$, it can be transformed by first applying the initial mapping $\sigma_1\define{\pi_1}^{-1}$  and then inserting any SWAP circuit implementing $\pi_2^{-1}$  between $\widetilde{C}_1$ and $\pi_2(\widetilde{C}_2)$.}

{As a concrete example, we show that the circuit in Fig.~\ref{fig:quantum_circuit} (and Sec.~\ref{sec:qct-ex}) can be reconstructed as a \qknob\ circuit.}

\noindent\textbf{The Circuit in Fig.~\ref{fig:quantum_circuit} as a \qknob\ Circuit:}
Let $V_1 \define\{0,1,2,3,4,5\}$, $E_1\define \{(0,1), (1,3), (2,4), (3,5)\}$ and $V_2\define \{1,2,3,4,5\}$, $E_2\define \{ (1,3), (2,4), (4,5)\}$. Then $G_i=(V_i,E_i)$ $(i=1,2)$ are two subgraphs of $\ag=Grid(3,2)$, see Fig.~\ref{fig:glink} for illustration. 

Let $\pi_2\define \pi_{0,1}$ be the permutation induced by swapping 0 and 1 in $\ag$. Permuting $G_2$ with $\pi_2$, and letting $G'$ be the union of $G_1$ and $\pi_2(G_2)$, $G'$ has edges $(0,1),(1,3), (2,4),(3,5), (0,3)$, and $(4,5)$ (see Fig.~\ref{fig:glink}(d)). Because $G'$ contains a 3-cycle $(0,1,3,0)$, it cannot be embedded in $\ag$. That is, we have a strong glink $\sglink{G_1}{\pi_2}{G_2}$.

\begin{figure}[tbp]
    \begin{tabular}{cccc}
    \includegraphics[width=0.1\textwidth]{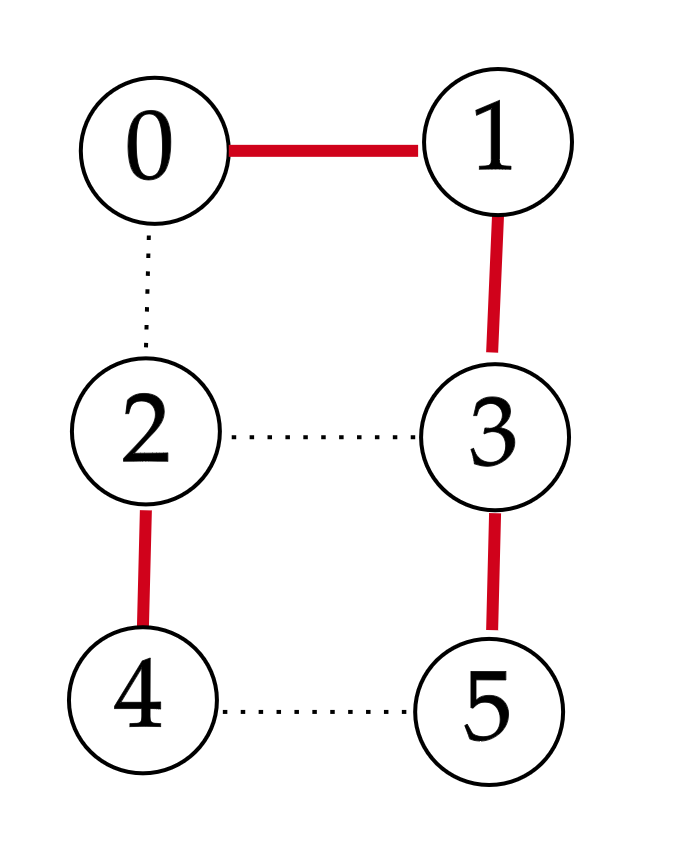} 
    &\hspace*{-2mm}
    \includegraphics[width=0.1\textwidth]{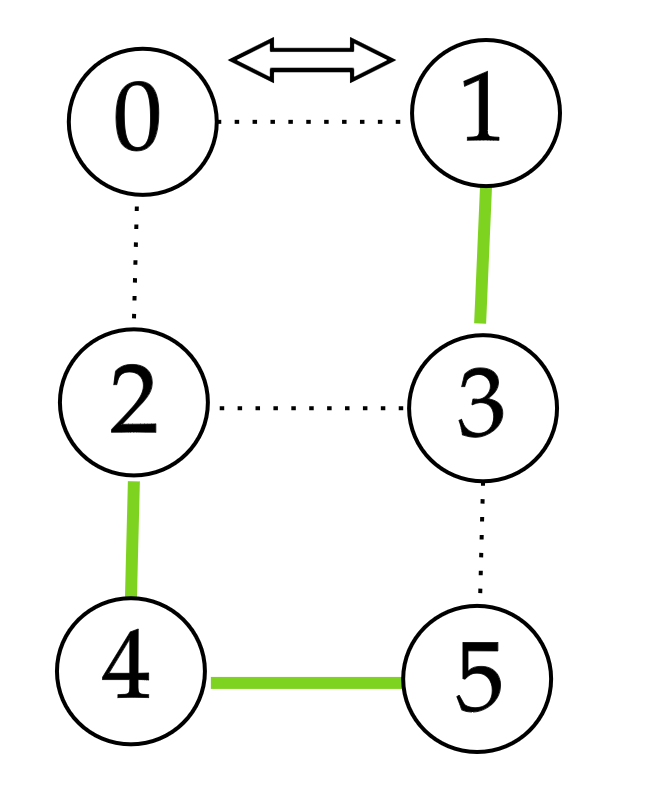} 
    &\hspace*{-2mm}
    \includegraphics[width=0.09\textwidth]{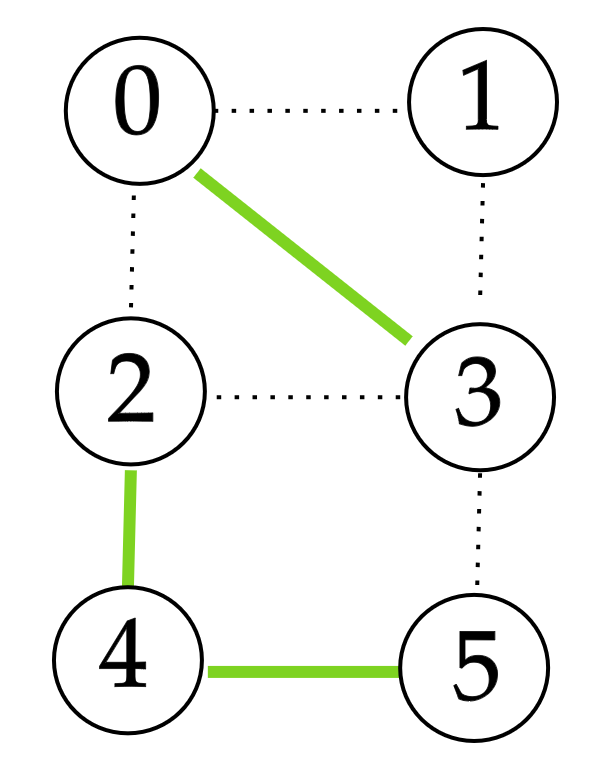} &\hspace*{-2mm}
    \includegraphics[width=0.09\textwidth]{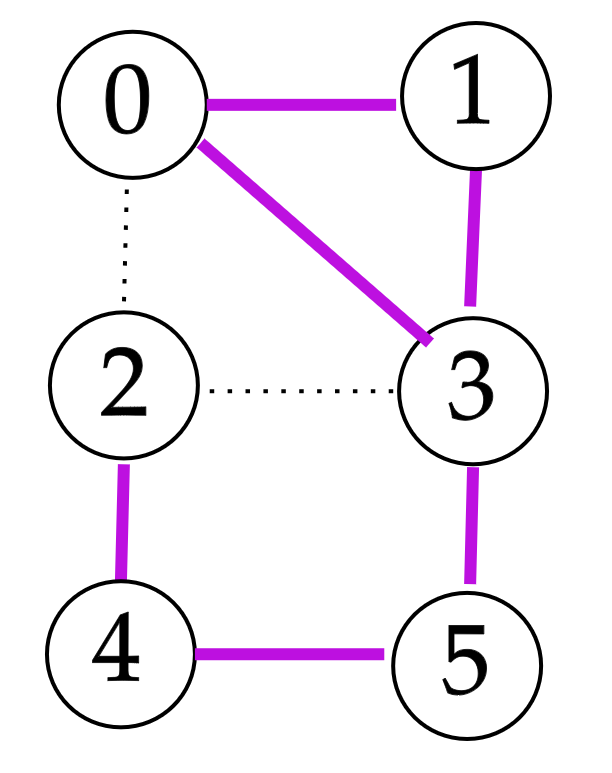} \\
    (a) & (b) & (c) & (d)
    \end{tabular}
    \caption{A strong glink $\sglink{G_1}{\pi_2}{G_2}$: (a) subgraph $G_1$; (b) subgraph $G_2$; (c) the graph $\pi_2(G_2)$ obtained by permuting $G_2$ with $\pi_2=\pi_{0,1}$; 
    (d) the union of $G_1$ and $\pi_2(G_2)$.}
    \label{fig:glink}
\end{figure}

Starting from $G_1$ and $G_2$, we construct two circuits, $\widetilde{C}_1$ and $\widetilde{C}_2$, whose interaction graphs are $G_1$ and $G_2$, respectively. For instance, let 
\begin{align*}
\resizebox{1\hsize}{!}{%
$\widetilde{C}_1 =  \Big[\la 2\ra , \la 0, 1\ra , \la 1\ra , \la 3, 5\ra , \la 3\ra , \la 1\ra , \la 3, 5\ra , \la 4\ra , \la 2\ra , \la 2, 4\ra , \la 1\ra , \la 5\ra $, $\la 1, 3\ra \Big]$
}
\\
\resizebox{1\hsize}{!}{%
$\widetilde{C}_2 = \Big[\la 0\ra , \la 1\ra , \la 2\ra , \la 3, 1\ra , \la 5\ra , \la 1, 3\ra , \la 1\ra , \la 1\ra , \la 4, 5\ra $, $\la 4\ra , \la 0\ra , \la 2, 4\ra \Big]$. 
}
\end{align*}
Then, $\widetilde{C}_1$ and $\widetilde{C}_2$ are two $\ag$-circuits. 

Permuting $\widetilde{C}_2$ with $\pi_2=\pi_{0,1}$ yields 
\begin{align*}
\resizebox{1\hsize}{!}{%
$\pi_2(\widetilde{C}_2) = \Big[\la 1\ra , \la 0\ra , \la 2\ra , \la 3, 0\ra , \la 5\ra , \la 0, 3\ra , \la 0\ra , \la 0\ra , \la 4, 5\ra , \la 4\ra , \la 1\ra , \la 2, 4\ra \Big].$
}
\end{align*}
Finally, we generate an arbitrary permutation $\pi_1$ and apply it to $\widetilde{C}_1+\pi_2(\widetilde{C}_2)$. For example, let $\pi_1=(2,1,5,4,3,0)$. We have a {\qknob} circuit \[C \define \pi_1\big(\widetilde{C}_1+\pi_2(\widetilde{C}_2)\big) = \pi_1(\widetilde{C}_1)+({\pi_1\circ\pi_2})(\widetilde{C}_2).\] Let $C_1 =\pi_1(\widetilde{C}_1)$ and $C_2=({\pi_1\circ\pi_2})(\widetilde{C}_2)$. Then  $C=C_1+C_2$.

The \qknob\ circuit $C$ is exactly the logical circuit shown in Fig.~\ref{fig:quantum_circuit} and studied in Sec.~\ref{sec:qct-ex}. As shown in Sec.~\ref{sec:qct-ex}, starting from the initial mapping $\sigma_1\define\pi_1^{-1} = (5,1,0,4,3,2)$, all gates in $C_1$ are executable and removed. After inserting a $\swap{0}{1}$, the mapping $\sigma_1=\pi_1^{-1}$ evolves to $\sigma_2\define ({\pi_1\circ\pi_2})^{-1} = \pi_2^{-1}\circ \pi_1^{-1} = \pi_2^{-1}\circ \sigma_1 = (5,0,1,4,3,2)$, which transforms {$\la 4,2\ra$ to $\la 3,1\ra$, $\la 3,0\ra$ to $\la 4,5\ra$, and $\la 5,3\ra$ to $\la 2,4\ra$}. All gates in $C_2$ are now executable.

\subsection{General \qknob\ Circuit Construction}

\begin{figure}[tbp]
    \centering
    \includegraphics[width=0.48\textwidth]{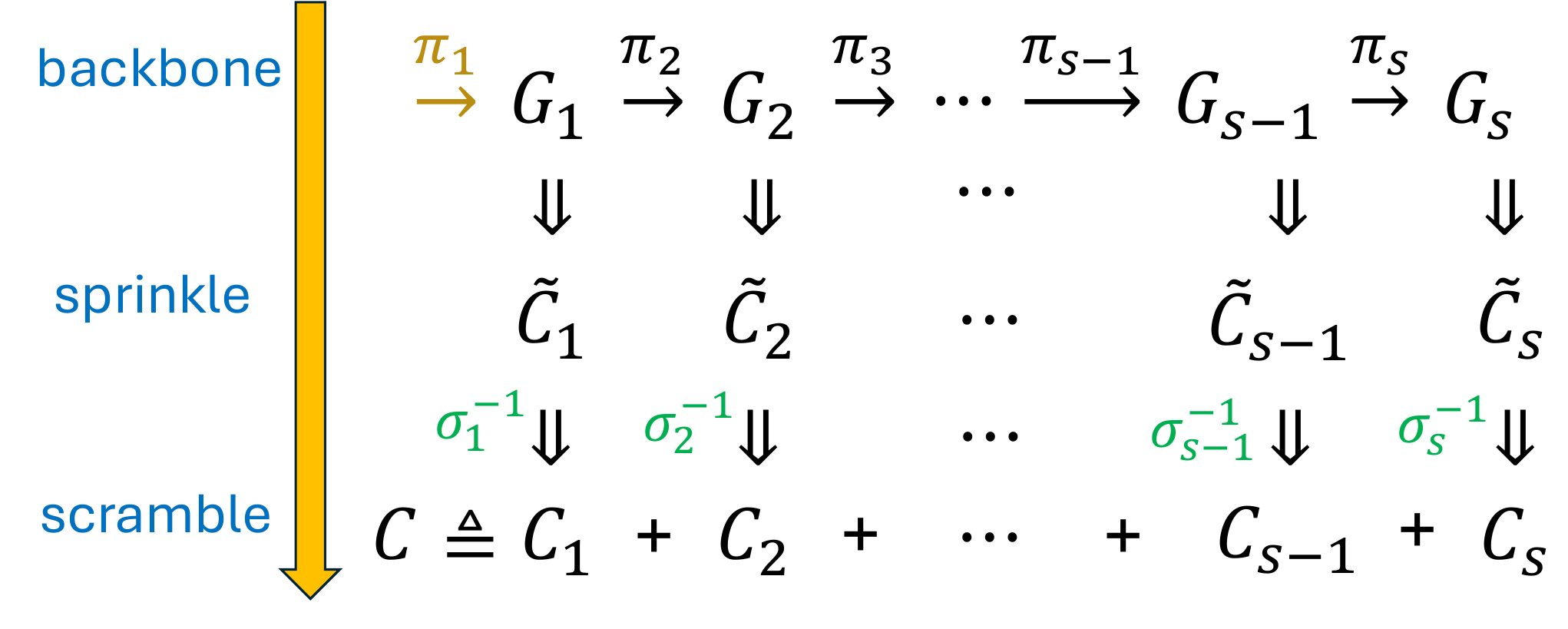}
    \caption{Construction of a \qknob\ circuit $C$ on an  architecture graph $\ag$, where  $\pi_i$, $G_i$, $\widetilde{C_i}$, $C_i$  represent, respectively, a permutation on $\ag$, a subgraph of $\ag$, an $\ag$-circuit, and a subcircuit of $C$. Additionally,     $\sigma_{i}^{-1}=\pi_1\circ\cdots\circ\pi_{i}$. The built-in transformation of $C$ has initial mapping $\sigma_1=\pi^{-1}_1$  and sequentially inserts SWAP circuits $\sem{\pi_i^{-1}}$ ($1\leq i\leq s$), where $\sem{\pi_i^{-1}}$ denotes any optimal SWAP circuit  implementing $\pi_i^{-1}$.}
    \label{fig:qknob}
\end{figure}

{To construct general \qknob\ circuits, we generate multiple subcircuits and connect them using glinks. This process relies on the notion of a glink chain, which is a sequence of glinks defining the connections between subcircuits and determining the cost of the built-in transformation for the resulting circuit.}
\begin{definition}[Glink Chain]\label{dfn:glinkchain}
A glink chain is a sequence 
\[G_1,\pi_2,G_2,\pi_3,\ldots, G_{s-1},\pi_{s},G_s\] 
such that $\glink{G_i}{\pi_{i+1}}{G_{i+1}}$ for each $1\leq i<s$. A glink chain is \emph{strong} if all its glinks are strong. 
\end{definition}

Given a glink chain $\la G_1,\pi_2,G_2,\ldots, \pi_{s},G_s \ra$ and a permutation $\pi_1$, we generate a \qknob\ circuit $C$ using a 3-phase construction process similar to \queko\ \cite{TanC21-queko} (see Fig.~\ref{fig:qknob}): 

1) For each $i$, generate a random circuit $\widetilde{C}_i$ such that its interaction graph is $G_i$.

2) Concatenate these circuits backward with the permutations: 
    \begin{align}\label{eq:quekno}
   \resizebox{0.9\hsize}{!}{%
    $C = \pi_1\bigg(\widetilde{C}_1 + \pi_2 \Big(\widetilde{C}_2 + \pi_3 \big( \cdots \pi_{s-1}(\widetilde{C}_{s-1} + \pi_{s}(\widetilde{C}_s))\cdots\big)\Big)\bigg)$.
    }
\end{align}
Note that each $\widetilde{C}_i$ constructed above is an $\ag$-circuit. 

The following theorem guarantees that $C$ has a built-in transformation with cost $\sum_{i=2}^s\swapnorm{\pi_{i}}$, which provides an upper bound for $C$'s optimal transformation cost. 

\begin{theorem}\label{thm:quekno}
Suppose $C$ is a \qknob\ circuit  as in Eq.~\ref{eq:quekno}. Then $C$ can be transformed into an $\ag$-executable circuit by: 
\begin{itemize}
    \item taking $\pi_1^{-1}$ as the initial mapping,
    \item inserting, from left to right, a SWAP circuit {$S_{i}$} that implements {$\pi_{i+1}^{-1}$} after (logically) executing and removing gates in $C_i\define  (\pi_1\circ\cdots\circ \pi_i)(\widetilde{C}_i)$ from $C$ for $1\leq i<s$.
\end{itemize}
This built-in transformation has cost of $\sum_{i=2}^s\swapnorm{\pi_{i}}$.
\end{theorem}
{The built-in transformation in Theorem~\ref{thm:quekno} corresponds precisely to the partition-and-permute form stated in Theorem~\ref{thm:qct}.} Indeed, letting $\sigma_1\define\pi_1^{-1}$ and $\sigma_{i+1}\define \pi_{i+1}^{-1}\circ \sigma_i$ for $1\leq i<s-1$ (cf. Fig.~\ref{fig:qknob}), we see that $\pi_{i+1}^{-1}=\sigma_{i+1}\circ \sigma_i^{-1}$, $\widetilde{C}_i = \sigma_i(C_i)$, and these  $C_i,\sigma_i$ satisfy the conditions given in Theorem~\ref{thm:qct}. 

There is a potential issue with using arbitrary permutation to scramble an $\ag$-circuit: computing the exact value of $\swapnorm{\pi}$ can be difficult. In our benchmark construction for SWAP count optimality, we require each of $\pi_2,\ldots,\pi_s$ to be derived from a small number of SWAPs. The first permutation, $\pi_1$, can still be arbitrary because it does not contribute to the transformation cost. To ensure that the transformation cost is close to the optimal one, we further require that the union of $G_i$ and $\pi_{i+1}(G_{i+1})$ is not embeddable in $\ag$ for each $1\leq i<s$. This implies that at least one SWAP must be inserted between two consecutive subcircuits $C_i$ and $C_{i+1}$. With these restrictions, the \qknob\ circuit $C$'s built-in SWAP cost is close to the optimal cost, although the exact cost generally remains difficult to compute.

While the above discussion focuses on benchmarks for evaluating SWAP count optimality, our approach can also generate benchmarks for evaluating depth optimality. For this purpose, we replace the permutations used above (except the first, which can still be arbitrary) with permutations implementable by a set of SWAP gates that do not share any qubits; that is, a SWAP circuit of depth 1. These SWAPs are often called \emph{parallel} SWAPs.

\section{Benchmark Design}\label{sec:design}

This section describes the detailed procedure for generating \qknob\ benchmarking circuits for $\ag=(\V,\E)$, the architecture graph of a quantum device.

\subsection{\qknob\ Construction Algorithm}

\begin{algorithm}
\caption{{\qknob} Circuits Construction}
\label{alg:quekno}
\small
\begin{algorithmic}[1]
\Require{An architecture graph $\ag$,  optimisation type $opt_\text{type}$, target cost $c$, permutation type $perm_\text{type}$, subgraph size $graph\_size$, and qubit gate ratio $\rho_\text{qbg}$}
\Ensure{A \qknob\ circuit $C$} 
\State Randomly generate a permutation $\pi_1$
\State Randomly generate a subgraph $G_1$ of $\ag$ respecting subgraph size
\State Randomly generate an $\ag$-circuit $\widetilde{C}_1$ with interaction graph $G_1$ respecting qubit gate ratio
\State $glinkChain \leftarrow G_1$
\State $cost \leftarrow 0$
\State $\ell \leftarrow 1$
\While{$cost < c$}
    \State Randomly generate a strong glink $\sglink{G_\ell}{\pi_{\ell+1}}{G_{\ell+1}}$ starting from the last subgraph $G_\ell$ of $glinkChain$  that respects  permutation type  and subgraph size
    \State Extend $glinkChain$ with $\sglink{G_\ell}{\pi_{\ell+1}}{G_{\ell+1}}$ 
    \State Randomly generate an $\ag$-circuit $\widetilde{C}_{\ell+1}$ with interaction graph $G_{\ell+1}$  respecting qubit gate ratio
    \If{$opt_\text{type}  = \text{`SWAP count'}$}
        \State $cost \leftarrow cost+\swapnorm{\pi_{\ell+1}}$ 
    \ElsIf{$opt_\text{type}  = \text{`depth'}$}
        \State $cost \leftarrow cost + 1$
    \EndIf
\State $\ell \leftarrow \ell + 1$
\EndWhile
\State
\resizebox{0.85\hsize}{!}{%
$C\leftarrow  \pi_1\bigg(\widetilde{C}_1 + \pi_2 \Big(\widetilde{C}_2 + \pi_3 \big( \cdots \pi_{\ell}(\widetilde{C}_{\ell} + \pi_{\ell+1}(\widetilde{C}_{\ell+1}))\cdots \big)\Big)\bigg)$ 
}
\end{algorithmic}
\end{algorithm}

We require the following subroutines (the parameters used there will be explained in detail in Sec.~\ref{sec:quekno-dim}):

\noindent (1) Randomly generate a subgraph $G$ of $\ag$ with the specified subgraph size (either `small' or `large').

\noindent (2) Given a subgraph $G$ of $\ag$, and a given qubit gate ratio $\rho_\text{qbg}$ (specifying the ratio between the numbers of one- and two-qubit gates), randomly generate an  $\ag$-circuit ($\widetilde{C}$) that respects $\rho_\text{qbg}$ and has $G$ as its interaction graph. 
    
\noindent (3) Given two subgraphs $G_1, G_2$ of $\ag$, randomly select a permutation $\pi$ of a given type such that $\sglink{G_1}{\pi}{G_2}$ (i.e., $\la G_1,\pi,G_2 \ra$ is a strong glink).  When the optimisation objective is `SWAP count', we have two permutation types, `opt1' and `opt2' {to be introduced in Sec.~\ref{sec:quekno-dim}). When the optimisation objective is `depth', the permutation is implemented by parallel SWAPs.

Using these subroutines, we can generate benchmarking circuits of the form shown in Eq.~\ref{eq:quekno}, whose built-in transformation costs are near-optimal. The pseudocode is described in Alg.~\ref{alg:quekno}. 

\vspace*{2mm}
From Theorem~\ref{thm:qct}, any transformation of an input circuit can be presented as in Eq.~\ref{eq:quekno}. This suggests that our benchmarks are general and representative. 

By Theorem~\ref{thm:quekno}, each circuit constructed using  Alg.~\ref{alg:quekno} has a built-in transformation. When the optimisation objective is  `SWAP count', we implement each permutation (other than $\pi_1$) using one or two SWAPs. Because we use strong glinks, at least one SWAP must be inserted between any two consecutive subcircuits. This implies that the built-in cost should be close to the optimal cost and thus provides us a known near-optimal transformation cost. When the optimisation objective is `depth', the situation is more complex. While each SWAP gate is implemented with three consecutive CNOTs, the depth increase from inserting a depth-1 SWAP circuit between $C_i$ and $C_{i+1}$ can be significantly greater than 3. Indeed, a single SWAP insertion can double the circuit depth \cite{li-iccad23sqgm}!  
To generate benchmarking circuits with a small depth ratio, we modify each $\widetilde{C}_i$ so that there are few or no free qubits in its last layer. This can be achieved by, for example, rearranging the gates in $\widetilde{C}_i$ or randomly inserting  one-qubit gates and some two-qubit gates that have already appeared.

\subsection{Dimensions and Parameters of \qknob} \label{sec:quekno-dim}

Our design has the following dimensions:

\noindent \textbf{Optimisation objective ($opt_\text{type}$)}: The target can be minimising the SWAP cost or the depth cost. SWAP cost counts the number of inserted SWAPs, whereas depth cost is the difference between the output and input circuit depths.

\noindent \textbf{Target transformation cost ($c$)}: If the optimisation objective is SWAP cost, we select $c$ from $\set{0,1,2,3,4,5,10,15,20,25}$; if the objective is depth cost, we select the cost (in swap layers) from $\set{1,2,3,4,5,10}$. We focus on near-term feasible circuits \cite{TanC21-queko} with depths usually no more than $80$. Because the circuit depth can increase by three or more when adding a layer or a sequence of SWAPs, considering SWAP cost below 25 and SWAP layers below 10 is reasonable. 
    
\noindent \textbf{Permutation type ($perm_\text{type})$}: For SWAP count optimisation, permutations (excluding the first, which does not contribute to transformation cost) are implemented by either a single SWAP or two consecutive SWAPs. If the permutation type is `opt1', each permutation is implemented by a single swap; if it is `opt2', each permutation is randomly implemented by either a single SWAP or two consecutive SWAPs. For depth optimisation, we implement each permutation (except the first) by a set of parallel SWAPs (i.e., a depth-1 SWAP circuit).

\noindent \textbf{Architecture graph ($\ag$)}: Our benchmark construction method can be applied on any near-term superconducting quantum devices. In this paper, we consider three representative quantum devices: IBM Q Tokyo (20 qubits), IBM Q Rochester (53 qubits), and Google's Sycamore (53 qubits), also  considered in \cite{TanC21-queko,LiZF21_fidls,Zhou+22_MCTS_Todaes}. Benchmarks designed for Rochester may not be near-optimal for Sycamore, and vice versa, because (i) some Sycamore subgraphs are not embeddable in Rochester, and (ii) a strong glink for Rochester might not be a strong glink for Sycamore. 
    
\noindent \textbf{Subgraph size ($graph\_size$)}: Subgraphs of the architecture graph are used for generating glinks. Larger subgraphs result in more gates in the subcircuits ($C_i$ in Alg.~\ref{alg:quekno}) of the benchmarking circuit. For the 20-qubit IBM Q Tokyo, the subgraphs in the generated glinks have approximately five edges  on average. For the  two 53-qubit devices, we offer two choices: small subgraphs (approximately eight edges) or large  subgraphs (approximately 16 edges).
    
\noindent \textbf{Qubit gate ratio ($\rho_\text{qbg}$) for evaluating depth optimality}: 
For depth optimality, the number (and distribution) of one-qubit gates in the input circuit can significantly affect the transformed circuit's depth. To reflect this, we introduce the `qubit gate ratio' parameter, $\rho_\text{qbg}=M_1/M_2$, in \qknob\ construction, where $M_1$ and $M_2$ denote the number of one- and two-qubit gates, respectively. Following \cite{TanC21-queko}, we consider two  ratios: the `QSE' ratio ($\rho_\text{qbg}= 2.55$) based on the random circuit used in Google's quantum supremacy experiment \cite{Arute+19_google_quantum_supremacy}, and the `TFL' ratio ($\rho_\text{qbg} = 1.5$) based on the Toffoli circuit.

\begin{table*}[]
    \caption{\qknob\  (top) and \queko\ benchmark sets (bottom)}
    \label{tab:benchmark_set}
    \centering
    \begin{tabular}{c|c|c|c|c|c}
    & benchmark set name & $perm_\text{type}$ & $graph\_size$ & $\rho_\text{qbg}$ & \#circuit \\ \hline 
    \multirow{6}{*}{\qknob} 
    &20Q\_gate\_tokyo& opt1 or opt2 &large &1.5& 100 $(\times 2)$\\
    &53Q\_gate\_Sycamore & opt1 or opt2  &small or large&1.5&100 $(\times 4)$\\
    &53Q\_gate\_Rochester & opt1 or opt2 &small or large&1.5&100 $(\times 4)$\\
    &20Q\_depth\_tokyo&parallel &large &1.5 or 2.55 & 60 $(\times 2)$\\
    &53Q\_depth\_Sycamore &parallel&small or large &1.5 or 2.55 & 60 $(\times 4)$\\
    &53Q\_depth\_Rochester &parallel&small or large &1.5 or 2.55 & 60 $(\times 4)$\\
 \hline
 \multirow{3}{*}{\queko}
&20Q\_bss\_tokyo &-&-&2.55& 90\\
&16Q\_bntf\_Aspen-4 &-&-&1.5 & 90\\
&54Q\_bntf\_Sycamore  &-&-&2.55& 90\\
    \end{tabular}
\end{table*}

\vspace*{2mm}
For each legal combination of these dimensions, we randomly generate 10 circuits.  In total, we have six sets of \qknob\ circuits, as shown at the top of Table~\ref{tab:benchmark_set}. For example, the circuit set `53Q\_gate\_Rochester' includes 100 circuits for each  parameter pair (permutation type, graph size), which can be used to evaluate the SWAP count optimality of QCT algorithms on IBM Q Rochester.

\section{Experiments and Evaluation}\label{sec:evaluation}
{This section evaluates the effectiveness of QKNOB circuits for benchmarking QCT algorithms. We compare on QKNOB and QUEKO benchmarks the performance of five state-of-the-art QCT algorithms: the transpiler of {\tket} \cite{Cowtan+19-tket} from Quantinuum, \sabre\ \cite{Li+19-sabre}, \sahs\ \cite{Zhou+20_SAHS}, \mcts\ \cite{Zhou+22_MCTS_Todaes}, and the Qiskit transpiler \texttt{StochasticSWAP}, focusing on SWAP count and depth optimality. Our experiments also highlight general trends across different architectures and optimisation objectives.} All our experiments were run on a laptop with i7-11800 CPU, 32 GB memory and RTX 3060 GPU.

\subsection{Details of the Compared QCT Algorithms}  
Qiskit provides multiple choices for both initial mapping and routing procedures.\footnote{The version of Qiskit  used in our evaluation is 0.33.0.} To facilitate comparison with \queko\ benchmarks, we choose \texttt{DenseLayout} and \texttt{StochasticSWAP} as the Qiskit transpiler to compare. \sabre\ was initially described in \cite{Li+19-sabre} and has recently been assembled in Qiskit. In this paper, we choose this Qiskit implementation of \sabre\ and select the advanced `lookahead' heuristics. 
The version of \sahs\ we use is from GitHub.\footnote{https://github.com/BensonZhou1991/circuittransform} The original \mcts\ algorithm \cite{ZhouFL20_MCTS_iccad} targeted SWAP count and was modified in \cite{Zhou+22_MCTS_Todaes} to address depth optimisation.\footnote{https://github.com/BensonZhou1991/MCTS-New} We call these two versions \mcts-size and \mcts-depth, respectively. In our evaluation of SWAP count optimality, we use \mcts-size; otherwise, we use \mcts-depth. The initial mappings used for \mcts\ are obtained by the Simulated Annealing method in \sahs\ \cite{Zhou+20_SAHS}. When evaluating \tket\ on \queko, Tan and Cong \cite{TanC21-queko} selected \texttt{GraphPlacement} for initial mapping, which might have led to favourable results for \tket, as \texttt{GraphPlacement} uses subgraph isomorphism and can find optimal solutions for some \queko\ circuits.  To facilitate comparison with \cite{TanC21-queko}, in our evaluation, we also choose to use \texttt{GraphPlacement}.\footnote{The version of \tket\ used for our evaluation is 0.17.0.}
 In addition, we disable all optimisation passes for fair comparisons. 
 
 As \mcts, \sabre, and \texttt{StochasticSWAP} are random algorithms, for each  circuit, we ran these algorithms five times and recorded the best value. This could improve the performance by 10\% - 20\%.

\subsection{Measuring SWAP Count and Depth Optimalities}
This work focuses on two key metrics for evaluating quantum circuit transformation: SWAP count optimality and depth optimality.

\textbf{SWAP Count Optimality} is measured using the CNOT gate ratio (`\textbf{cx ratio}'), which is the ratio of the number of CNOT gates in the transformed (output) circuit to the number of CNOT gates in the original (input) circuit. Formally, the CNOT gate ratio is defined as:
\begin{equation}\label{eq:rho_gate}
\rho_{\text{cx}} = \frac{\text{number of CNOT gates in the output circuit}}{\text{number of CNOT gates in the input circuit}}.
\end{equation}
Here, we assume that each inserted SWAP gate is decomposed into three consecutive CNOT gates in the output circuit. 

A smaller value of $\rho_{\text{cx}}$ indicates a transformation requires fewer SWAP insertions. Note that $\rho_{\text{cx}}\geq 1$, and equality occurs when no SWAP is inserted during the transformation.

\textbf{Depth optimality} is measured using the \textbf{depth ratio} (denoted as $\rho_{\text{depth}}$), which is the ratio of the circuit depth after transformation to the depth of the original circuit. Formally, he depth ratio is defined as: 
\begin{equation}\label{eq:rho_depth}
    \rho_{\text{depth}} = \frac{\text{depth of the output circuit}}{\text{depth of the input circuit}}.
\end{equation}
Similarly, the smaller the value of $\rho_{\text{depth}}$, the closer the transformation is to optimal in terms of depth. Note that $\rho_{\text{depth}}\geq 1$, with equality achieved when the transformed circuit has the same depth as the original. 


{These two metrics collectively provide a comprehensive evaluation of quantum circuit transformations, capturing both gate-level efficiency and overall circuit execution complexity. We also use $\rho_{\text{cx}}-1$ and $\rho_{\text{depth}}-1$ to measure the gate overhead per CNOT gate and the depth overhead per layer, respectively, relative to the input circuit.
}

\subsection{Summary of Evaluation Results}

\begin{figure}[htb]
    \centering
    \begin{tabular}{c}
    \includegraphics[width=0.48\textwidth]{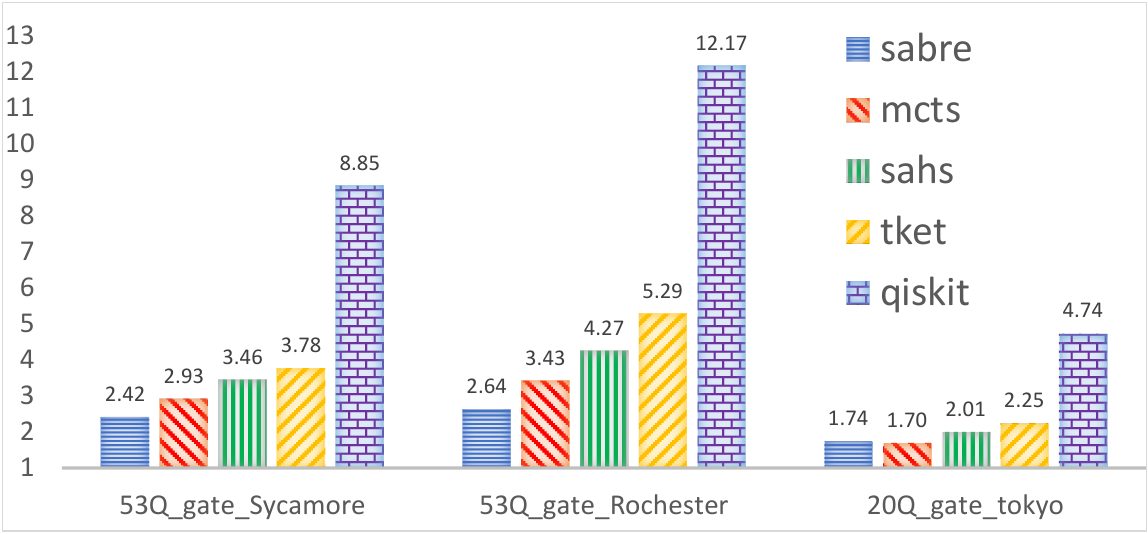}\\
    \includegraphics[width=0.48\textwidth]{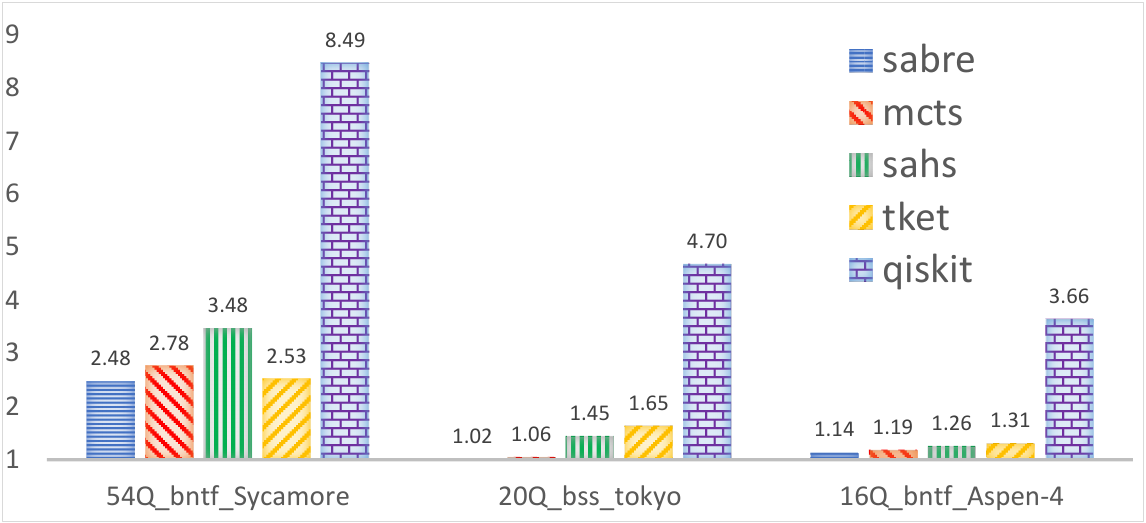}
    \end{tabular}
    \caption{SWAP count optimality performance of QCT algorithms on the \qknob\ (top) and \queko\ (bottom) benchmark sets, where the $y$-axes denote the average cx ratios $\rho_{\text{cx}}$ (cf. Eq.~\ref{eq:rho_gate}); lower values are better. }
    \label{fig:gate_opt_sum}
\end{figure}

\begin{figure}[htb]
    \centering
    \begin{tabular}{c}
    \includegraphics[width=0.48\textwidth]{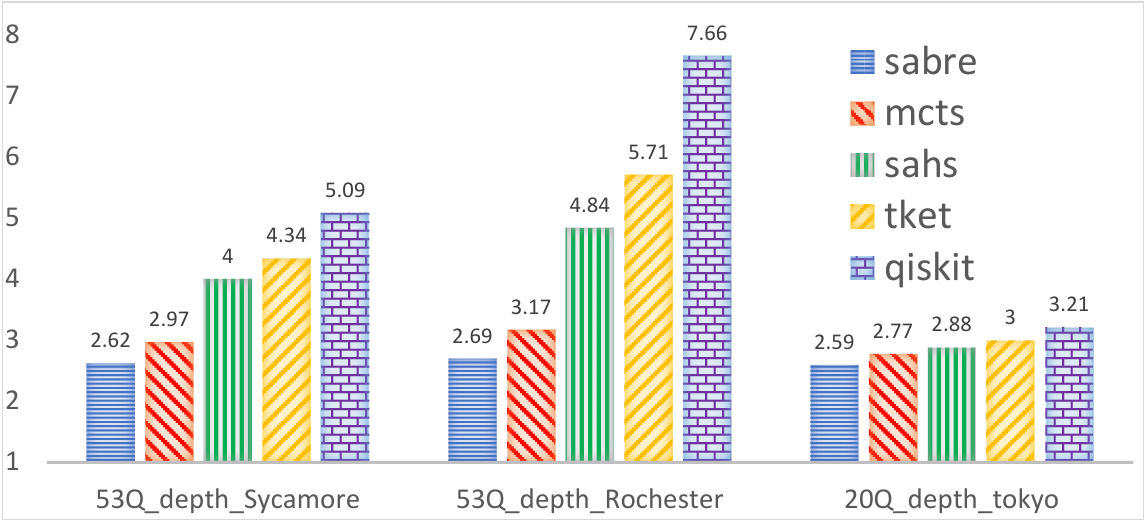}\\
    \includegraphics[width=0.48\textwidth]{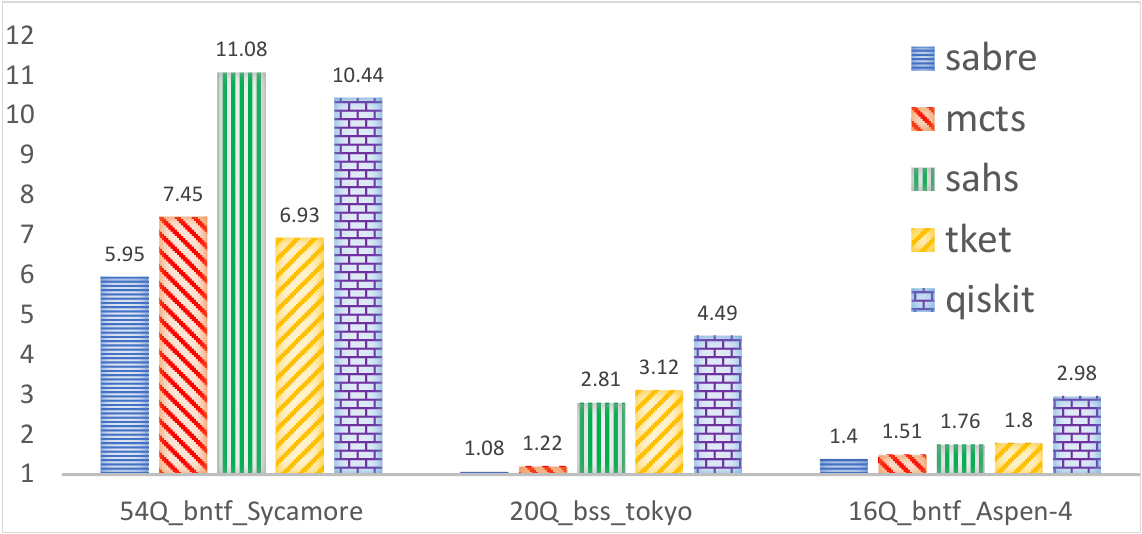}
    \end{tabular}
    \caption{Depth optimality performance of QCT algorithms on the \qknob\ (top) and \queko\ (bottom) benchmark sets, where the $y$-axes denote the average depth ratios $\rho_{\text{depth}}$ (cf. Eq.~\ref{eq:rho_depth}); lower values are better. }
    \label{fig:depth_opt_sum}
\end{figure}


\textbf{For SWAP count optimality}, we consider three \qknob\ benchmark sets: `53Q\_gate\_Sycamore', `53Q\_gate\_Rochester', `20Q\_gate\_Tokyo', and three \queko\ sets `54Q\_bntf\_Sycamore', `20Q\_bss\_Tokyo', and `16Q\_bntf\_Aspen-4'. As no `20Q\_bntf\_Tokyo' benchmark set is provided in the GitHub website of \queko,\footnote{https://github.com/tbcdebug/QUEKO-benchmark} we replace it with `20Q\_bss\_Tokyo', which are benchmarks with depth from 100 to 900 for scaling study. We also run benchmarks in the  `16Q\_bntf\_Aspen-4' on the 20-qubit IBM Q Tokyo. In addition, when running on benchmarks in `54Q\_bntf\_Sycamore', we use the ideal architecture graph of Sycamore where the bad node, as well as its connections, is restored. The evaluation results are summarised in Fig.~\ref{fig:gate_opt_sum}, from which we can see:\\
\noindent 1) Qiskit's \texttt{StochasticSWAP} performs significantly worse  across all six benchmark sets. \\
\noindent 2) \sabre\ performs the best on all but one benchmark sets; its average cx ratio (cf. Eq.~\ref{eq:rho_gate}) on the \qknob\ benchmark set `20Q\_gate\_Tokyo' is only slightly worse than that of \mcts-size (1.74 vs. 1.70). Its performance on the two 53-qubit \qknob\ benchmark sets is conspicuously ($>17\%$) better than the other algorithms.

\textbf{For depth optimality},  we evaluate, in addition to the three \queko\ sets, \qknob\ sets `53Q\_depth\_Sycamore', `53Q\_depth\_Rochester', and `20Q\_depth\_Tokyo'. The evaluation results are summarised in Fig.~\ref{fig:depth_opt_sum}, from which we can see:\\
\noindent 1) Qiskit's \texttt{StochasticSWAP} performs clearly worst on all but `54Q\_bntf\_Sycamore'; its performance on `54Q\_bntf\_Sycamore' is the second-worst.
\\
\noindent 2) \sabre\ performs best on all  benchmark sets; its performance on the two 53-qubit \qknob\ benchmark sets is at least $\geq 10\%$ better than that of any other algorithm.

\vspace*{2mm}
When comparing performances on the \qknob\ benchmarks across different architectures, Figures~\ref{fig:gate_opt_sum} and~\ref{fig:depth_opt_sum} reveal the following general trends: 
\[\sabre\! \prec \mcts\! \prec   \sahs\!  \prec  \text{t}\!\ket{\mathrm{ket}}\!  \prec \texttt{StochasticSWAP},\] where $\prec$ denotes `better than'. Note that \sabre\ $\prec$ \mcts\ is violated only on `20Q\_gate\_Tokyo', where the cx ratios of \sabre\ and \mcts\ are 1.74 and 1.70, respectively. Furthermore, the cx/depth ratios on Sycamore are generally lower than those on Rochester, reflecting Sycamore's greater qubit connectivity.

\begin{remark}
    \sabre\ is faster than all tested algorithms except \texttt{StochasticSWAP}. To investigate whether repeated executions of \texttt{StochasticSWAP} could improve results, we compared the algorithms on the ten `53Q\_depth\_Rochester\_large\_opt\_10\_2.55' circuits. The results are presented below, where `Ratio' denotes the depth ratio; \mcts\ denotes \mcts-depth; and \texttt{SSx5}, \texttt{SSx100}, \texttt{SEx5}, and \texttt{SEx100} indicate repeating \texttt{StochasticSWAP} or \sabre\ 5 or 100 times, respectively.
    
    \begin{table}[h]
    \centering
    \setlength{\tabcolsep}{2pt} 
    \small 
    \begin{tabular}{c|ccccccc}
    Alg.\! & \sahs & \mcts & \tket &\! \texttt{SSx5} &\! \texttt{SSx100} & \texttt{SEx5} & \texttt{SEx100} \\ \hline
    Ratio & 6.80 & 4.03 & 7.89 & 8.07 & 7.16 & 4.11 & \textbf{3.35} \\
    Time (s) & 356 & 2755 & 977 & 76 & 1150 & 159 & 2690
    \end{tabular}
    \end{table}
\end{remark}

\subsection{Comparing \qknob\ with \queko} 

As pointed out in the introduction, \queko\ circuits are designed with optimal transformations that incur zero costs. While this simplifies benchmark construction, it limits \queko's ability to faithfully evaluate QCT algorithms that rely on subgraph isomorphism for initial mapping. For instance, \fidls\ \cite{LiZF21_fidls} achieves optimal transformations on all \queko\ `16Q\_bntf\_Aspen-4' and `20Q\_bss\_Tokyo' circuits, yet its performance on \qknob\ `20Q\_gate\_Tokyo' is comparable to that of \tket, far from being optimal.

Fig.~\ref{fig:gate_opt_sum} (bottom) provides further insights. Both \sabre\ and \mcts-size exhibit near-optimal performance on \queko\ benchmark set `20Q\_bss\_Tokyo' (1.02 vs. 1.06), where the optimality score is 1. This near-optimal performance is unexpected, as neither algorithm uses subgraph isomorphism for initial mapping. This suggests that \queko\ fails to reveal these QCT algorithms' true performance. In contrast, on \qknob\ `20Q\_gate\_Tokyo', the cx ratios for \sabre\ and \mcts-size are 1.74 and 1.70, respectively.  A similar phenomenon is also observed in Fig.~\ref{fig:depth_opt_sum} (bottom): \sabre\ and \mcts-depth achieve near-optimal depth ratios (1.08 vs. 1.22) on \queko\  `20Q\_bss\_Tokyo' but exhibit much higher depth ratios on \qknob\   `20Q\_depth\_Tokyo' (2.59 vs. 2.77). 

Another notable observation is \tket's inconsistent performance on the \queko\ and \qknob\ benchmarks. As shown in Figures~\ref{fig:gate_opt_sum} and~\ref{fig:depth_opt_sum}, \tket\ significantly outperforms \mcts\ and \sahs\ on \queko\   `54Q\_bntf\_Sycamore' but performs worse than these algorithms on the four 53Q \qknob\ benchmark sets. This discrepancy likely stems from \tket's \texttt{GraphPlacement} mapping pass, which can find optimal initial mappings for certain \queko\ benchmarks. 

In summary, \qknob\ addresses the limitations of \queko\ and demonstrates its effectiveness in faithfully evaluating QCT algorithms. This capability stems from \qknob's construction methodology (cf. Alg.~\ref{alg:quekno}) and the theoretical properties established in Theorems~\ref{thm:qct} and ~\ref{thm:quekno}.

\begin{figure*}[tbp]
    \centering
    \begin{tabular}{cc}
   \includegraphics[width=0.45\textwidth]{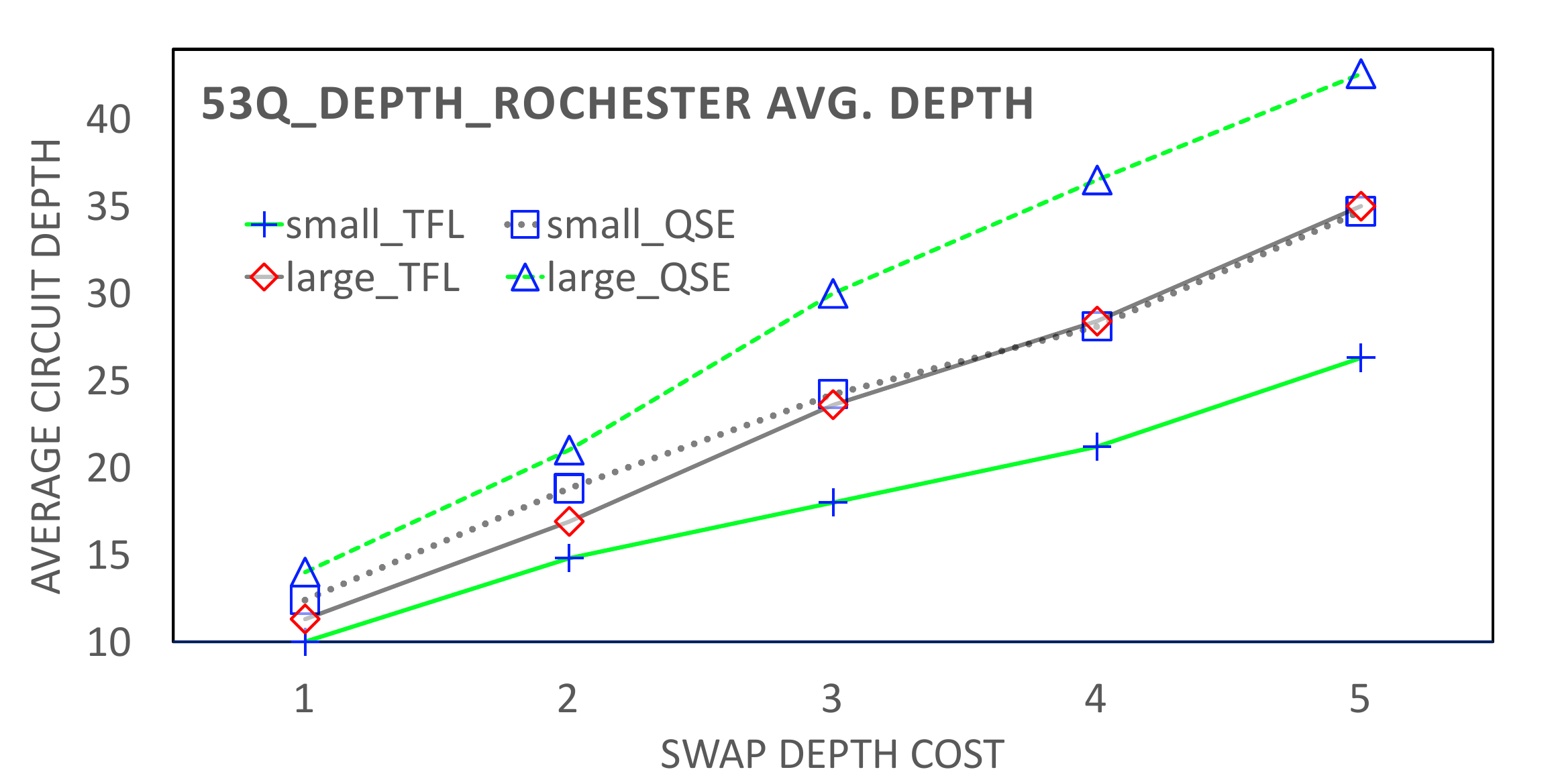}
    &
    \includegraphics[width=0.45\textwidth]{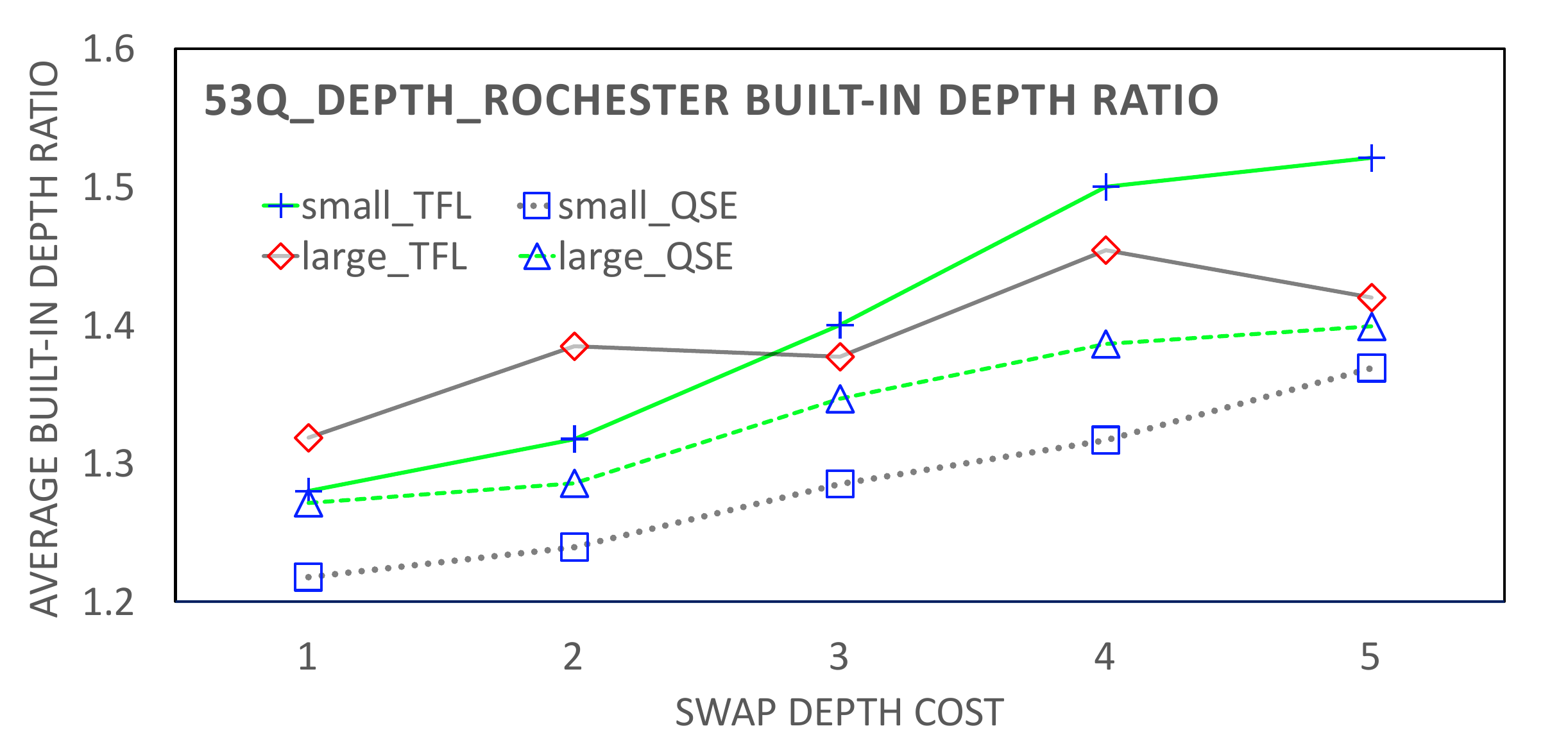} 
    \\
    (a) & (b) \\
    \includegraphics[width=0.5\textwidth]{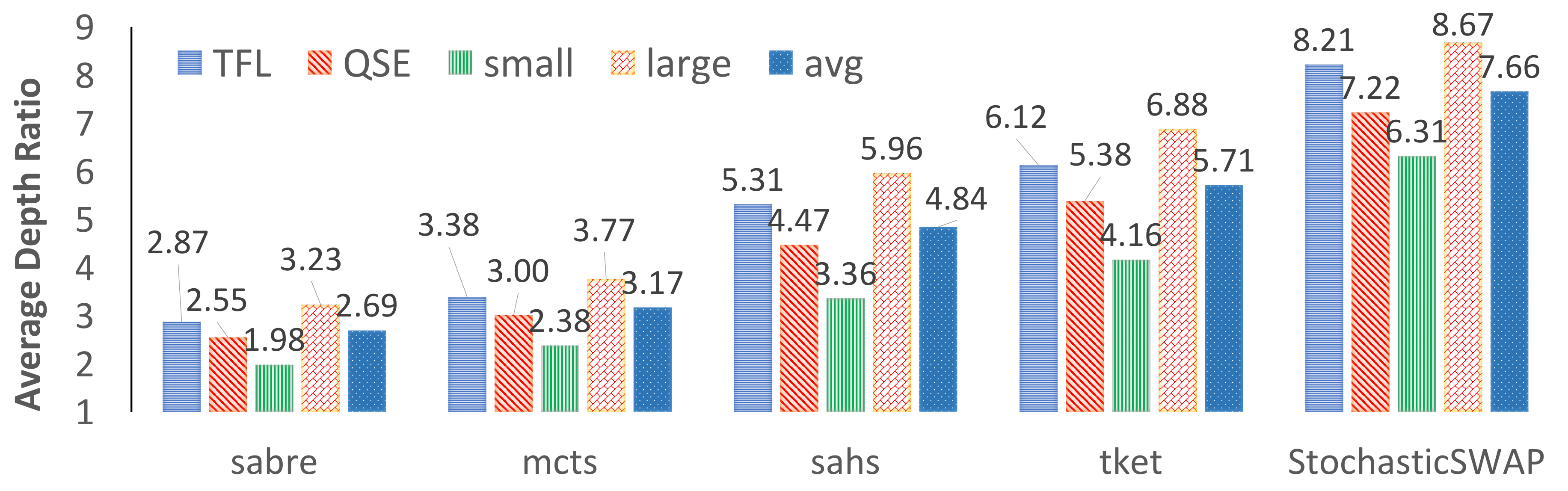} 
&
    \includegraphics[width=0.45\textwidth]{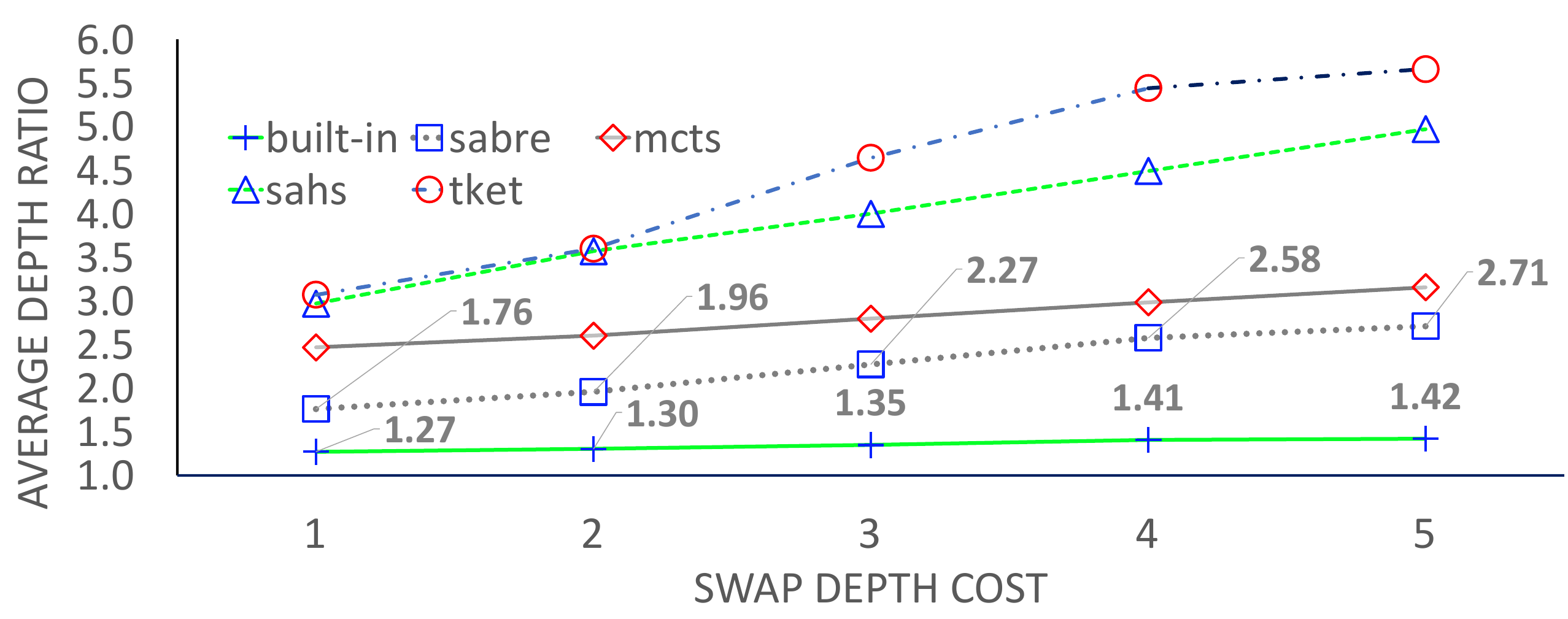}
\\
(c) & (d)
    \end{tabular} 
    \caption{Information and evaluation results for \qknob\  `53Q\_depth\_Rochester' circuits on IBM Q Rochester: (a) average circuit depth; (b) average built-in depth ratios; (c) performance comparison among four subsets (`TFL', `QSE', `small', `large') and the full set (`avg'); (d) evaluation results. In (a), (b), and (d), the $x$-axis denotes the target SWAP depth cost, and each point represents the average value over ten \qknob\ circuits with the same parameters.}
    \label{fig:rochester_depth_info}
\end{figure*}

\subsection{Benchmark Information}
{This subsection provides a detailed examination of the \qknob\ benchmarks used in our evaluation.} Due to space limitation, we narrow our discussion to the Rochester benchmark sets for depth optimality and analyse how the depths of the benchmarking circuits and the depth ratios associated with their built-in transformations (henceforth `built-in depth ratios') vary with target depth cost.

In Fig.~\ref{fig:rochester_depth_info}(a,b), we examine four sets of \qknob\ `53Q\_depth\_Rochester' circuits, characterised by their graph size (`large' or `small') and qubit gate ratio (`TFL', corresponding to $\rho_\text{qbg}=1.5$, or `QSE', corresponding to $\rho_\text{qbg}=2.55$). The target SWAP depth costs are within the values 1 to 5. For clarity, data for circuits with a target SWAP depth cost of 10---which have average depths between 40 and 80---are omitted from the plots. 

Fig.~\ref{fig:rochester_depth_info}(a) demonstrates that the average depths of these circuits scale linearly with the target depth cost. Circuits with a `large'  graph size and a  `QSE' qubit gate ratio exhibit larger depths compared to their counterparts with a `small'  graph size and a qubit gate ratio.  This trend highlights the impact of graph size and qubit gate ratio on circuit depth.

Fig.~\ref{fig:rochester_depth_info}(b) compares the built-in depth ratios across these circuits. It shows that `TFL' circuits generally have greater built-in depth ratios than their `QSE' counterparts. This discrepancy is partially attributable to `QSE' circuits containing more one-qubit gates, which increase the overall depth. Additionally,  the figure reveals that the depth ratios of \qknob\ circuits scale effectively (at most linearly) with increasing target SWAP depth costs.

{These observations validate the design of \qknob\ benchmarks, demonstrating their scalability and ability to represent diverse circuit properties, which are essential for evaluating depth optimality in QCT algorithms.}

\subsection{Factors That Affect Optimality}

The design of \qknob\ circuits highlights three critical factors that influence QCT algorithms's SWAP count (or depth) optimality: target cost, permutation type (or qubit gate ratio), and subgraph size. 

For large devices like Rochester, the subgraphs used in glinks are classified as either `small' (approximately 8 edges) or `large' (approximately 16 edges). Circuits with `small' subgraphs naturally contain fewer gates. Additionally, \qknob\ circuits with the `TFL' qubit gate ratio have fewer one-qubit gates but a comparable number of two-qubit gates compared to those with the `QSE' ratio.

Fig.~\ref{fig:rochester_depth_info}(c) reveals two key observations:
\begin{enumerate}
    \item \textbf{Significant Effect of Qubit Gate Ratio:} The qubit gate ratio has a noticeable impact on QCT algorithms, often exceeding a $10\%$ difference; 
    \item \textbf{Greater Impact of Subgraph Size:} The subgraph size has an even larger effect, with QCT algorithms achieving significantly lower depth ratios (often more than 50\% lower) on circuits with small subgraphs. This is partly because smaller subgraphs lead to fewer total gates and reduced circuit depth. \emph{Circuits with large subgraphs are inherently more challenging to transform.} For instance, \sabre\ achieves an average depth ratio of 3.23 on circuits with large subgraphs, which is 63\%  higher than its average depth ratio of 1.98 on circuits with small subgraphs.
\end{enumerate} 

Fig.~\ref{fig:rochester_depth_info}(d) shows that the performance of QCT algorithms scales well with increasing target depth costs in the construction of \qknob\ circuits.

\section{Further Discussion}
\label{sec:discussion}

\noindent\textbf{Scalability:} The proposed \qknob\ benchmark construction method operates in time linear to the number of gates. However, it requires a subgraph isomorphism check to determine whether a glink is strong (line~8 of Alg.~\ref{alg:quekno}), which is not polynomial in the number of qubits. This limitation is manageable for two reasons. First, near-term quantum devices typically have no more than several thousand qubits, for which subgraph isomorphism can be checked efficiently. Second, the construction method often only needs to disprove a subgraph isomorphism, which is computationally easier than proving it. For example, it takes about 200 seconds to generate a 900-qubit \qknob\ circuit on $Grid(30,30)$ with permutation type `opt2', graph\_size 450, qubit gate ratio 1.5, and target transformation cost 10.

\noindent\textbf{Generality and Theoretical Guarantees:} 
The construction method described in Sec.~\ref{sec:quekno-theory} is general. Theorems~\ref{thm:qct} and~\ref{thm:quekno} guarantee that for any circuit $C$ and any transformation of $C$, $C$ can be reconstructed as a \qknob\ circuit with the given transformation as its  built-in transformation. Consequently, every \qknob\ circuit has a built-in transformation and an associated known cost, which serves as an upper bound for the optimal transformation cost. 

\noindent\textbf{Addressing the Gap Between Built-In and Optimal Costs:}  Using arbitrary permutations in Eq.~\ref{eq:quekno} makes it difficult to estimate the \emph{gap} between the built-in and optimal costs. For \qknob\ circuits designed to  benchmark SWAP count optimality, permutations generated by one or two consecutive SWAPs help to narrow the gap. In contrast, constructing benchmarking circuits for depth optimality (cf. the paragraph before Sec.~\ref{sec:quekno-dim}) is less theoretically rigorous. Despite this limitation, the  constructed \qknob\ circuits exhibit relatively small built-in depth ratios. For instance, the `53Q\_depth\_Rochester' circuits with a target SWAP layer cost of 5 achieves an average built-in depth ratio of 1.42 (cf.~Fig.~\ref{fig:rochester_depth_info}(d)), translating to an average depth overhead of 0.42---only about a quarter of that incurred by \sabre\ (2.71-1=1.71)!

\noindent\textbf{Limitations of the Framework:} The \qknob\ framework does not account for commutation rules, synthesis optimisations, gate absorption, or remote CNOTs. Consequently, evaluating QCT algorithms employing such techniques \cite{Itoko+19_commutation,Liu+22_not_all_swap,Xie+21_commutativity,TanC21-gate_absorption,Niemann+21_combine_remote_cnot} or  mitigating crosstalk errors \cite{Xie+21_commutativity} may introduce bias.

\noindent\textbf{Adapting to Hardware Characteristics:} Real quantum devices exhibit spatial and temporal variability in qubit reliability and connectivity strengths. The \qknob\ construction method can be adapted to reflect these characteristics by leveraging 
calibration data to select subgraphs comprising highly reliable qubits and edges. Similarly, permutations for generating strong glinks can prioritise highly reliable edges. This approach effectively reverses the variability-aware qubit mapping strategy in \cite{Murali+19,TannuQ19}, tailoring benchmarks to specific hardware characteristics. 

\section{Conclusion}
\label{sec:conclusion}

In this work, we introduced \qknob, a novel benchmarking framework for quantum circuit transformation (QCT) that features circuits with built-in transformations and near-optimal costs. Unlike existing benchmarks, such as \queko, \qknob\ reflects the general QCT process, enabling the generation of representative and versatile benchmarking circuits. {Researchers can also use our open-source benchmark construction algorithm to customise benchmarks for specific devices or algorithm evaluations.}



Through detailed experiments, we demonstrated \qknob’s fairness and comprehensiveness in evaluating QCT algorithms, including those leveraging subgraph isomorphism (e.g., \tket). Notably, \sabre\ emerged as the best-performing algorithm on large devices, highlighting its potential for scaling to future quantum hardware. {However, our results also reveal a significant gap between the performance of state-of-the-art QCT algorithms and the built-in transformations of \qknob\ circuits. For instance, on medium-depth circuits, even the best-performing algorithms incur more than four times the overhead of the built-in transformations (cf.Fig.\ref{fig:rochester_depth_info}(d)). This underscores the need for further innovation to narrow this gap.}


In summary, \qknob\ provides a robust and scalable framework for benchmarking QCT algorithms, advancing the evaluation of both existing and emerging algorithms. While it highlights the strengths of current approaches, it also identifies critical areas for improvement. Future research could focus on developing more optimal QCT algorithms, incorporating hardware-aware adaptations, and leveraging \qknob\ to explore new strategies for enhancing real-world performance.

\section*{Code and Benchmark Availability}
{Our algorithm (implemented in Python 3) and benchmarks are publicly available through the MIT license 
at \url{https://github.com/ebony72/quekno}. 
}

\appendix[Proofs]
\begin{proof}[Proof of Lemma~\ref{lem:permcomp}]
 Let $S_1 = \big( \swap{p_1}{q_1},\ldots,\swap{p_c}{q_c} \big)$ and $S_2 = \big( \swap{p_{c+1}}{q_{c+1}},\ldots,\swap{p_{c+d}}{q_{c+d}} \big)$, where $c,d\geq 1$ and each $(p_j,q_j)$ is an edge in $G$. By definition, we have {$\pi_1 = \pi_{p_c,q_c}\circ\cdots\circ\pi_{p_1,q_1}$ and $\pi_2 = \pi_{p_{c+d},q_{c+d}}\circ\cdots\circ\pi_{p_{c+1},q_{c+1}}$. It is now clear that $\pi_2\circ \pi_1=\pi_{p_{c+d},q_{c+d}}\circ\cdots\circ\pi_{p_{c+1},q_{c+1}}\circ \pi_{p_{c},q_{c}}\circ\cdots\circ\pi_{p_1,q_1}$} is implemented by $S_1+S_2= \big( \swap{p_1}{q_1},\ \ldots,\ \swap{p_c}{q_c},\  \swap{p_{c+1}}{q_{c+1}},\ \ldots$, $\swap{p_{c+d}}{q_{c+d}} \big)$.
\end{proof}

\begin{proof}[Proof of Lemma~\ref{lem:2}]
 This follows directly from Definition~\ref{dfn:pcirc}.
\end{proof}

\begin{proof}[Proof of Theorem~\ref{thm:qct}]
By assumption, $C$ has a transformation with cost $c$. That is, we can transform $C$ into an executable circuit using an initial (logical to physical) mapping $\sigma_1$ and $c$ SWAP gates $\swap{p_1}{q_1},\ldots,\swap{p_c}{q_c}$. Let $C_1$ be $C$'s subcircuit containing all gates executable under $\sigma_1$. Due to optimality, $C_1$ cannot be empty. Removing all gates in $C_1$ and inserting the SWAP gates one by one until some gates in the remainder of $C$ are executable under the current mapping. Let $S_1$ denote this SWAP circuit and $\pi_2$ denote the inverse of the permutation implemented by $S_1$. Then $\sigma_2 \define \pi_2^{-1}\circ\sigma_1$ is the current (logical to physical) mapping. Continuing the above procedure until all gates in $C$ are executed, we partition $C$ into $1\leq s\leq c+1$ nonempty subcircuits $C_1,\ldots, C_s$ and partition the $c$ SWAP gates into $s-1$ nonempty SWAP circuits $S_1,\ldots,S_{s-1}$ such that $S_i$ is inserted between $C_i$ and $C_{i+1}$. Let $\pi_{i+1}$ be the inverse of the permutation implemented by $S_i$  and let $\sigma_{i+1} \define \pi_{i+1}^{-1}\circ\sigma_i$. It can be proved inductively that $C_i$ is executable under $\sigma_i$ for $1\leq i\leq s$. Clearly, $C$ has the form as shown in Eq.~\ref{eq:circpart}. Because $\pi_{i+1}^{-1}=\sigma_{i+1}\circ\sigma_i^{-1}$ and $S_i=\sem{\pi_{i+1}^{-1}} $ for each $1\leq i\leq s-1$,  the transformed circuit has the form shown in Eq.~\ref{eq:physical-circuit-final}.
\end{proof}

\begin{proof}[Proof of Lemma~\ref{lem:glink}]
 Taking the identity mapping $id_\V$ as the initial mapping, gates in $\widetilde{C}_1$ can be executed and removed from $C$ directly. The current mapping remains to be $id_\V$. If we insert a SWAP circuit that implements {$\pi^{-1}$, the current logical to physical mapping becomes $\pi^{-1}$}, which can execute $\pi(\widetilde{C}_2)$ because $\pi^{-1}(\pi(\widetilde{C}_2))=\widetilde{C}_2$ is an $\ag$-circuit. 
\end{proof}

\begin{proof}[Proof of Theorem~\ref{thm:quekno}]
 For each $1\leq i\leq s$, let $\sigma_i$ denote the permutation $(\pi_1\circ \cdots\circ \pi_i)^{-1}$. Then $C$ is represented in the form shown in Eq.~\ref{eq:circpart}. Following the analysis given in Sec.~\ref{sec:qct} or directly by Lemma~\ref{lem:glink}, we can show the correctness of the transformation. Because the cost of this transformation is $c \define \sum_{i=2}^s {\swapnorm{\pi_i^{-1}} =\sum_{i=2}^s \swapnorm{\pi_i}}$, we know that the optimal cost is no greater than $c$.
\end{proof}

\bibliographystyle{IEEEtran}
\bibliography{qct-lsj,qct24plus}

\newpage
\begin{IEEEbiography}[{\includegraphics[width=1in,height=1.25in,clip,keepaspectratio]{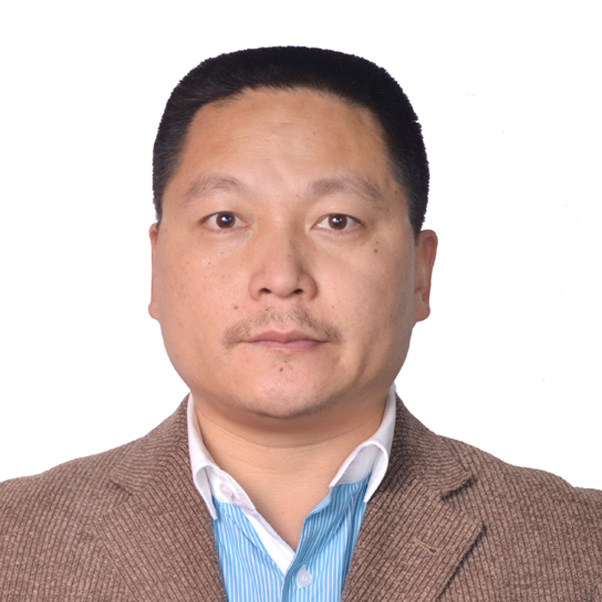}}]{Sanjiang Li}
 received his B.Sc. and Ph.D. degrees in mathematics from Shaanxi Normal University, in 1996, and Sichuan University, in 2001, respectively. He is a professor in Centre for Quantum Software and Information,  University of Technology Sydney (UTS). 
 His research interests are mainly in knowledge representation, artificial intelligence, and quantum circuit compilation. 
\end{IEEEbiography}

\begin{IEEEbiography}[{\includegraphics[width=1in,height=1.25in,clip,keepaspectratio]{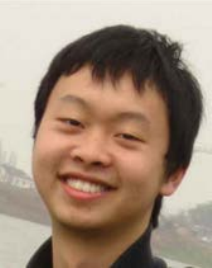}}]{Xiangzhen Zhou} received his BE (2010) and PhD (2022) degrees in engineering from Nanjing Normal University and Southeast University, China. From 2018 to 2020, he was a visiting student at the Centre for Quantum Software and Information, University of Technology Sydney. 
He is currently a lecturer at Nanjing Tech University, China. His research interests include quantum computing and quantum circuit optimisation.
\end{IEEEbiography}

\begin{IEEEbiography}[{\includegraphics[width=1in,height=1.25in,clip,keepaspectratio]{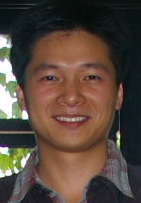}}]{Yuan Feng}
received his BS degree in Applied Mathematics and PhD degree in Computer Science from Tsinghua University, in 1999 and 2004 respectively. He is currently a professor at the Department of Computer Science and Technology, Tsinghua University, China. His research interests include quantum programming theory, quantum information and quantum computation, and probabilistic systems.
\end{IEEEbiography}

\end{document}